\documentclass{article}

\usepackage{microtype}
\usepackage{graphicx}
\usepackage{subfigure}
\usepackage{booktabs} 

\usepackage[accepted]{icml2024}

\usepackage{multirow}
\usepackage{amssymb,amsmath,amsthm,amsfonts}
\usepackage{bm}
\usepackage{textgreek}
\usepackage{mathtools}
\usepackage{enumitem}
\usepackage{authblk}
\usepackage{graphicx}
\usepackage[font=small]{caption}
\usepackage[linesnumbered,ruled,vlined,boxed,algo2e]{algorithm2e}
\usepackage{physics}
\usepackage{footnote}
\usepackage{xcolor}
\usepackage{mathrsfs}
\usepackage{bbm}
\usepackage{makecell}
\usepackage[ruled,vlined,linesnumbered]{algorithm2e}
\usepackage{algorithmic}
\renewcommand{\algorithmicrequire}{\textbf{Input: }}
\renewcommand{\algorithmicensure}{\textbf{Output: }}

\definecolor{mydarkblue}{rgb}{0,0.08,0.45}
\usepackage[colorlinks,linkcolor=mydarkblue,citecolor=mydarkblue,anchorcolor=blue]{hyperref}

\newtheorem{theorem}{Theorem}
\newtheorem{definition}{Definition}
\newtheorem{lemma}{Lemma}

\newtheorem{corollary}{Corollary}
\newtheorem{remark}{Remark}

\newtheorem{problem}{Problem}

\newcommand{\specialcell}[2][c]{%
  \begin{tabular}[#1]{@{}c@{}}#2\end{tabular}}

\newcommand{\eqn}[1]{(\ref{eqn:#1})}
\newcommand{\eq}[1]{(\ref{eq:#1})}

\newcommand{\thm}[1]{\hyperref[thm:#1]{Theorem~\ref*{thm:#1}}}
\newcommand{\cor}[1]{\hyperref[cor:#1]{Corollary~\ref*{cor:#1}}}
\newcommand{\defn}[1]{\hyperref[defn:#1]{Definition~\ref*{defn:#1}}}
\newcommand{\lem}[1]{\hyperref[lem:#1]{Lemma~\ref*{lem:#1}}}
\newcommand{\prop}[1]{\hyperref[prop:#1]{Proposition~\ref*{prop:#1}}}
\newcommand{\assum}[1]{\hyperref[assum:#1]{Assumption~\ref*{assum:#1}}}
\newcommand{\fig}[1]{\hyperref[fig:#1]{Figure~\ref*{fig:#1}}}
\newcommand{\tab}[1]{\hyperref[tab:#1]{Table~\ref*{tab:#1}}}
\newcommand{\algo}[1]{\hyperref[algo:#1]{Algorithm~\ref*{algo:#1}}}
\renewcommand{\sec}[1]{\hyperref[sec:#1]{Section~\ref*{sec:#1}}}
\newcommand{\append}[1]{\hyperref[append:#1]{Appendix~\ref*{append:#1}}}
\newcommand{\fac}[1]{\hyperref[fac:#1]{Fact~\ref*{fac:#1}}}
\newcommand{\lin}[1]{\hyperref[lin:#1]{Line~\ref*{lin:#1}}}
\newcommand{\prob}[1]{\hyperref[prob:#1]{Problem~\ref*{prob:#1}}}

\def\>{\rangle}
\def\<{\langle}

\newcommand{\inner}[2]{\left\langle #1,\, #2 \right\rangle}
\newcommand{\vect}[1]{\ensuremath{\mathbf{#1}}}
\newcommand{\x}{\ensuremath{\mathbf{x}}}
\newcommand{\y}{\ensuremath{\mathbf{y}}}
\newcommand{\z}{\ensuremath{\mathbf{z}}}
\newcommand{\g}{\ensuremath{\mathbf{g}}}

\newcommand{\R}{\mathbb{R}}

\newcommand{\E}{\mathbb{E}}

\newcommand{\xout}{\x^{\mathrm{out}}}

\DeclareMathOperator{\poly}{poly}

\DeclareMathOperator{\Var}{Var}

\renewcommand{\x}{\vect{x}}
\renewcommand{\v}{\vect{v}}
\newcommand{\xref}{\vect{x}_{\mathrm{ref}}}
\newcommand{\w}{\vect{w}}

\newcommand{\0}{\mathbf{0}}

\def\:{\hbox{\bf:}}

\SetKwInput{KwInput}{Input}
\SetKwInput{KwParameter}{Parameters}                %
\SetKwInput{KwOutput}{Output}              %
\SetKwInput{KwReturn}{Return}

\let\oldnl\nl
\newcommand{\nonl}{\renewcommand{\nl}{\let\nl\oldnl}}

\usepackage[textsize=tiny]{todonotes}

\icmltitlerunning{Quantum Algorithms and Lower Bounds for Finite-Sum Optimization}

\begin{document}

\twocolumn[
\icmltitle{Quantum Algorithms and Lower Bounds for Finite-Sum Optimization}



\icmlsetsymbol{equal}{*}

\begin{icmlauthorlist}
\icmlauthor{Yexin Zhang}{equal,4}
\icmlauthor{Chenyi Zhang}{equal,3}
\icmlauthor{Cong Fang}{5,6}
\icmlauthor{Liwei Wang}{5,6}
\icmlauthor{Tongyang Li}{1,2}

\end{icmlauthorlist}

\icmlaffiliation{1}{Center on Frontiers of Computing Studies, Peking University, China}
\icmlaffiliation{2}{School of Computer Science, Peking University, China}
\icmlaffiliation{3}{Computer Science Department, Stanford University, USA}
\icmlaffiliation{4}{School of Electronics Engineering and Computer Science, Peking University, China}
\icmlaffiliation{5}{National Key Lab of General Artificial Intelligence, School of
Intelligence Science and Technology, Peking University}
\icmlaffiliation{6}{Institute for Artificial Intelligence, Peking University}

\icmlcorrespondingauthor{Tongyang Li}{tongyangli@pku.edu.cn}
\icmlcorrespondingauthor{Liwei Wang}{wanglw@pku.edu.cn}
\icmlcorrespondingauthor{Cong Fang}{fangcong@pku.edu.cn}

\icmlkeywords{Machine Learning, ICML}

\vskip 0.3in
]



\printAffiliationsAndNotice{\icmlEqualContribution} 

\begin{abstract}
Finite-sum optimization has wide applications in machine learning, covering important problems such as support vector machines, regression, etc. In this paper, we initiate the study of solving finite-sum optimization problems by quantum computing. Specifically, let $f_1,\ldots,f_n\colon\R^d\to\R$ be $\ell$-smooth convex functions and $\psi\colon\R^d\to\R$ be a $\mu$-strongly convex proximal function. The goal is to find an $\epsilon$-optimal point for $F(\x)=\frac{1}{n}\sum_{i=1}^n f_i(\x)+\psi(\x)$. We give a quantum algorithm with complexity $\tilde{O}\big(n+\sqrt{d}+\sqrt{\ell/\mu}\big(n^{1/3}d^{1/3}+n^{-2/3}d^{5/6}\big)\big)$,\footnote{The $\tilde{O}$ and $\tilde{\Omega}$ notations omit poly-logarithmic terms, i.e., $\tilde{O}(f)=O(f\poly(\log f))$ and $\tilde{\Omega}(f)=\Omega(f\poly(\log f))$.}  improving the classical tight bound  $\tilde{\Theta}\big(n+\sqrt{n\ell/\mu}\big)$. We also prove a quantum lower bound
$\tilde{\Omega}(n+n^{3/4}(\ell/\mu)^{1/4})$ when $d$ is large enough. Both our quantum upper and lower bounds can extend to the cases where $\psi$ is not necessarily strongly convex, or each $f_i$ is Lipschitz but not necessarily smooth. In addition, when $F$ is nonconvex, our quantum algorithm can find an $\epsilon$-critial point using $\tilde{O}(n+\ell(d^{1/3}n^{1/3}+\sqrt{d})/\epsilon^2)$ queries.
\end{abstract}


\section{Introduction}
In machine learning, especially supervised learning, it is common that the overall loss function can be written as a sum of loss functions at each data point. In particular, let $f_1,\ldots,f_n\colon\R^d\to\R$ be a sequence of functions, our goal is to find an approximate minimum of the function
\begin{align}\label{eqn:finite-sum}
F(\x)\coloneqq f(\x)+\psi(\x),
\end{align}
where $f(\x)$ satisfies
\begin{align}\label{eqn:finite-sum-noprox}
f(\x)=\frac{1}{n}\sum_{i=1}^n f_i(\x)
\end{align}
and $\psi(\x)$ is a known convex function, sometimes referred to as the \textit{proximal function}. For instance, given $n$ training data $(\x_1,y_1),\ldots,(\x_n,y_n)$ where $\x_i\in\R^{d-1}, \y_i\in\R$ for each $i\in[n]$, if $f_i$ is the square loss $f_i(\w,b)=(\w^{\top}x_i-y_i-b)^{2}$ for $\w\in\R^{d-1}$ and $b\in\R$, 
then Eq.~\eqn{finite-sum-noprox} gives linear least squares regression of these data. In general, such finite-sum optimization problems arises in many places in machine learning, statistics, and operations research, such as support vector machines, logistic regression, Lasso, etc.

Finite-sum optimization has been widely studied in previous literature given their wide applicability. \citet{zhang2004solving} solved finite-sum optimization by randomly selecting an index $i\in[n]$ and applying stochastic gradient descent (SGD). More efficient algorithms apply variance reduction, including SAG by~\citet{roux2012stochastic}, SDCA by~\citet{shalev2013stochastic}, SVRG by~\citet{johnson2013accelerating}, and also many other works~\cite{zhang2013linear,defazio2014saga,xiao2014proximal,allen2016improved,allen2016variance,reddi2016stochastic,shalev2016sdca,allen2017katyusha,allen2018katyusha,lin2018catalyst}.

More recently, quantum computing is rapidly advancing and there is also significant interest in quantum algorithms for continuous optimization faster than classical counterparts. This began with quantum algorithms for solving linear and semidefinite programs~\cite{brandao2016quantum,vanApeldoorn2018SDP,brandao2017SDP,vanApeldoorn2017quantum,casares2020quantum,kerenidis2018interior}, then for general convex optimization~\cite{vanApeldoorn2020optimization,chakrabarti2020optimization}, and now there are also quantum algorithms for slightly nonconvex problems~\cite{li2022quantum,chen2023quantum}, escaping saddle points in nonconvex landscapes~\cite{zhang2021quantum,childs2022quantum}, and also finding global minima in some special classes of nonconvex problems~\cite{liu2023quantum,leng2023quantum}. On the other hand, quantum lower bounds for convex optimization~\cite{garg2020no,garg2021near} and nonconvex optimization~\cite{gong2022robustness,zhang2023quantum} are also established. However, these optimization results mainly investigate the block-box setting with a single function $f(x)$, and the perspective of quantum algorithms for finite-sum stochastic optimization is widely open.

\paragraph{Contributions.}
In this work, we initiate the study of the quantum analogue of the standard finite-sum optimization problem. We assume quantum access to a finite-sum oracle, or access to a \emph{quantum finite-sum oracle} for brevity, that allows us to query the gradients of different $f_i$'s at the same time in \emph{quantum superpositions}.

\begin{definition}[Quantum finite-sum oracle]\label{defn:Of}
For an $F\colon\R^d\to\R$ and sub-functions $f_1,\ldots,f_n\colon\R^d\to\R$ satisfying the finite-sum structure in Eq.~\eqn{finite-sum-noprox}, its quantum finite-sum oracle $O_F$ is defined as\footnote{
In this paper, whenever we access a quantum oracle $U$, it is a unitary operation and we also have access to its corresponding inverse operation denoted as $U^\dag$, which satisfies $U^\dag U=UU^\dag=I$. This is a standard assumption, explicitly or implicitly employed, in previous research on quantum algorithms, see e.g., \cite{brassard2002quantum,cornelissen2022near,sidford2023quantum}.}
\begin{align}\label{eqn:Of}
\hspace{-2mm}O_F\ket{\x}\otimes\ket{i}\otimes\ket{0}\to\ket{\x}\otimes\ket{i}\otimes\ket{\nabla f_i(\x)}\ \ \forall i\in[n].
\end{align}
\end{definition}
Here, the Dirac notation $\ket{\cdot}$ denotes input or output registers made of qubits that may present as \textit{quantum superpositions}. Specifically, for an $\x\in\R^{d}$ and a coefficient vector $\vect{c}\in\mathbb{C}^n$ with $\sum_{i \in [n]} |c_i|^2=1$, the quantum register could be in the quantum state $\ket{\x}\sum_{i \in [n]} c_i\ket{i}\otimes\ket{\nabla f_i(\x)}$, which is a quantum superposition over all these $n$ sub-functions simultaneously.\footnote{Note that we store real numbers in these quantum registers, which solicit encoding from binary numbers to real numbers. Similar to classical encoding, we represent a number as a binary string and store each bit in a quantum bit. For example, for the real number 3.25, we encode by its binary representation 11.01, and then represent this number using four qubits. This method is commonly used when real numbers need to be manipulated mathematically in quantum algorithms. Of course, this encoding method also incurs errors when encoding real numbers, just like classical encoding, thus we may consider precision beforehand when using this encoding method. Basic operations can be implemented with constant time overhead~\cite{nielsen2000book}.} If we measure this quantum state, we will get  $\nabla f_i(\x)$ with probability $|c_i|^2$. If we query $O_F$ with no superposition over the first two registers, it collapses to a classical finite-sum oracle that returns the gradient of a specific $f_i$ at $\x$.

Next, we characterize the objective functions using the following properties:
\begin{itemize}
\vspace{-2mm}
\item $L$-Lipschitzness: For any $\x,\y\in\R^d$,
\[\|f(\x)-f(\y)\|\leq L\|\x-\y\|.\]

\vspace{-2mm}
\item $\mu$-strong convexity: For any $\x,\y\in\R^d$,
\[f(\x)-f(\y)\leq\nabla f(\x)^{\top}(\x-\y)-\frac{\mu}{2}\|\x-\y\|^{2}.\]

\vspace{-2mm}
\item $\ell$-smoothness: For any $\x,\y\in\R^d$,
\[\|\nabla f(\x)-\nabla f(\y)\|\leq \ell\|\x-\y\|.\]
\end{itemize}

\paragraph{Finite-sum convex optimization.}
In this work, we systematically investigate the quantum analogue of the finite-sum convex optimization problem.

\begin{problem}[Quantum finite-sum convex optimization (QFCO)]\label{prob:QFCO}
In the quantum finite-sum convex optimization (QFCO) problem we are given query access to a quantum finite-sum oracle $O_F$ for a convex function $F\colon\R^d\to\R$ satisfying \eqn{finite-sum}. The goal is to output an expected $\epsilon$-optimal point $\x^*\in\R^d$ satisfying $\E[F(\x^*)] \leq \inf_{\x}F(\x) + \epsilon$.
\end{problem}

We consider the following four cases of \prob{QFCO}: 

\begin{enumerate}
\item $\psi$ is $\mu$-strongly convex, each $f_i$ is convex and $\ell$-smooth, and $F(\0)-F^*\leq\Delta$. Example: ridge regression, elastic net regularization;
\item $\psi$ is not necessarily strongly convex, each $f_i$ is convex and $\ell$-smooth, and $F$ achieves its minimum at $\x^*$ with $\|\x^*\|\leq R$. Example: logistic regression, Lasso;
\item $\psi$ is $\mu$-strongly convex, each $f_i$ is convex and $L$-Lipschitz (not necessarily smooth), and $F(\0)-F^*\leq\Delta$. Example: $\ell_2$-norm support vector machine;
\item $\psi$ is not necessarily strongly convex, each $f_i$ is convex and $L$-Lipschitz (not necessarily smooth), and $F$ achieves its minimum at $\x^*$ with $\|\x^*\|\leq R$. Example: $\ell_1$-norm support vector machine.
\end{enumerate}

We develop quantum algorithms for all the four cases of \prob{QFCO} respectively in \sec{convex-algorithms}. The query complexities of these algorithms are summarized in the following.

\begin{theorem}[Informal version of \thm{Q-Katyusha} and \cor{HOOD-quantum}]
There exist four quantum algorithms that solve all the cases of \prob{QFCO}, respectively, with the following query complexities:
\begin{itemize}
\vspace{-3mm}
\item Case 1: $\tilde{O}\big(n+\sqrt{d}+\sqrt{\frac{\ell}{\mu}}\big(n^{1/3}d^{1/3}+n^{-2/3}d^{5/6}\big)\big)$;
\item Case 2: $\tilde{O}\big(n+\sqrt{d}+R\sqrt{\frac{\ell}{\epsilon}}\big(n^{1/3}d^{1/3}+n^{-2/3}d^{5/6}\big)\big)$;
\item Case 3: $\tilde{O}\big(n+\sqrt{d}+\frac{L}{\sqrt{\lambda\mu}}\big(n^{1/3}d^{1/3}+n^{-2/3}d^{5/6}\big)\big)$;
\item Case 4: $\tilde{O}\big(n+\sqrt{d}+\frac{LR}{\epsilon}\big(n^{1/3}d^{1/3}+n^{-2/3}d^{5/6}\big)\big)$.
\end{itemize}
\end{theorem}

Compared to~\citet{allen2017katyusha}, our quantum algorithms achieve a better query complexity when the dimension $d$ is relatively small. Prior to our work, \citet{ozgul2023stochastic} studied non-logconcave sampling from a distribution $\pi(\x)\propto\exp(-\beta f(\x))$ where $f(\x)=\frac{1}{n}\sum_{i=1}^n f_i(\x)$ is a finite sum of functions. As $\pi(\x)$ is larger when $f(\x)$ is smaller, measuring such a distribution can in principle solve the finite-sum optimization problem. However, compared to our result, there algorithms have dimension factor being at least $\Omega(d)$, and it is also not clear how large $\beta$ should be to reach the same criteria in optimization. As far as we know, we give the first quantum algorithm with speedup for finite-sum convex optimization. A more detailed comparison can be found in \tab{main}.

\begin{table*}[ht]
\centering
\caption{Comparisons between algorithms and lower bounds for finite-sum convex optimization, both in classical and quantum settings. The columns from left to right cover the four cases respectively.
$n$ is the number of functions $f_i$, $\epsilon$ is the error to the optimal value, $\Delta$ is an upper bound on the difference between $F(\0)$ and the optimal value, and $R$ is an upper bound on the norm of the optimum $\x^{*}$.}
\resizebox{\textwidth}{!}{
\begin{tabular}{|c|c|c|c|c|}
\hline
 & \multicolumn{2}{c|}{$\ell$-Smooth} & \multicolumn{2}{c|}{$L$-Lipschitz}  \\\cline{2-5}
 & $\mu$-Strongly Convex & Convex & $\mu$-Strongly Convex & Convex \\\hline\hline
\specialcell{Classical Upper Bound\\\cite{allen2017katyusha}} & $
 \tilde{O}\left(n+\sqrt{\frac{n\ell}{\mu}}\right)$ &
 $\tilde{O}\left(n+ R\sqrt{\frac{n\ell}{\epsilon}}\right)$ &
 $\tilde{O}\left(n+ L\sqrt{\frac{n}{\mu\epsilon}}\right)$ &
 $\tilde{O}\left(n+ LR\frac{\sqrt{n}}{\epsilon}\right)$ \\ \hline
\specialcell{Classical Lower Bound\\\cite{woodworth2016tight}} &
${\Omega} \left(n+\sqrt{\frac{n\ell}{\mu}}\log \frac{\Delta}{\epsilon}\right)$ &
${\Omega} \left(n+R\sqrt{\frac{n\ell}{\epsilon}}\right)$ &
${\Omega} \left(n+\frac{\sqrt{n}L}{\sqrt{\mu \epsilon}}\right)$&
${\Omega} \left(n+\frac{\sqrt{n}L R}{\sqrt{ \epsilon}}\right)$\\ \hline\hline
\specialcell{Quantum Upper Bound\\\textbf{(this work)}} & $\tilde{O}\big(n+\sqrt{d}+\sqrt{\frac{\ell}{\mu}}\big(n^{\frac{1}{3}}d^{\frac{1}{3}}+n^{-\frac{2}{3}}d^{\frac{5}{6}}\big)\big)$ & $\tilde{O}\big(n+\sqrt{d}+R\sqrt{\frac{\ell}{\epsilon}}\big(n^{\frac{1}{3}}d^{\frac{1}{3}}+n^{-\frac{2}{3}}d^{\frac{5}{6}}\big)\big)$ & $\tilde{O}\big(n+\sqrt{d}+\frac{L}{\sqrt{\lambda\mu}}\big(n^{\frac{1}{3}}d^{\frac{1}{3}}+n^{-\frac{2}{3}}d^{\frac{5}{6}}\big)\big)$ & $\tilde{O}\big(n+\sqrt{d}+\frac{LR}{\epsilon}\big(n^{\frac{1}{3}}d^{\frac{1}{3}}+n^{-\frac{2}{3}}d^{\frac{5}{6}}\big)\big)$ \\ \hline
\specialcell{Quantum Lower Bound\\\textbf{(this work)}} & $\tilde\Omega\left(n+n^{\frac{3}{4}}\left(\frac{\ell}{\mu}\right)^{\frac{1}{4}}\right)$ & $\tilde{\Omega}\left(n+n^{\frac{3}{4}}\left(\frac{\ell}{\epsilon}\right)^{\frac{1}{4}} R^{\frac{1}{2}}\right)$ & $\tilde{\Omega}\left(n+n^{\frac{3}{4}}\left(\frac{1}{\epsilon\mu}\right)^{\frac{1}{4}} L^{\frac{1}{2}}\right)$ & $\tilde{\Omega}\left(n+n^{\frac{3}{4}}\left(\frac{LR}{\epsilon}\right)^{\frac{1}{2}}\right)$ \\ \hline
\end{tabular}}
\label{tab:main}
\end{table*}

We also establish quantum complexity lower bounds for the four cases of \prob{QFCO} respectively in \sec{convex-lowerbound}, with the specific forms as follows. These quantum lower bounds confirm that the speedup for finite-sum optimization by quantum computing is at most polynomial in all parameters.

\begin{theorem}[Informal version of \cor{case1-lowerbound}, \cor{case2-lowerbound}, \cor{case4-lowerbound}, and \cor{case3-lowerbound}]
There exist four families of functions corresponding to the four cases such that
when $d$ is large enough, any quantum algorithm that finds an $\epsilon$-optimal point requires  the following query complexities:
\vspace{-2mm}
\begin{itemize}
\item Case 1: $\tilde\Omega\left(n+n^{\frac{3}{4}}\left(\frac{\ell}{\mu}\right)^{\frac{1}{4}}\right)$;
\item Case 2: $\tilde{\Omega}\left(n+n^{\frac{3}{4}}\left(\frac{\ell}{\epsilon}\right)^{\frac{1}{4}} R^{\frac{1}{2}}\right)$;
\item Case 3: $\tilde{\Omega}\left(n+n^{\frac{3}{4}}\left(\frac{1}{\epsilon\mu}\right)^{\frac{1}{4}} L^{\frac{1}{2}}\right)$;
\item Case 4: $\tilde{\Omega}\left(n+n^{\frac{3}{4}}\left(\frac{LR}{\epsilon}\right)^{\frac{1}{2}}\right)$.
\end{itemize}
\end{theorem}

\paragraph{Finite-sum critical point computation.}
Additionally, we develop a quantum algorithm for finding critical points, i.e., points with small gradients, of (possibly) non-convex functions in the finite-sum setting.

\begin{problem}[Quantum finite-sum critical point computation (QFCP)]\label{prob:QFCP}
In the quantum finite-sum critical point computation (QFCP) problem we are given query access to a quantum finite-sum oracle $O_F$ for a (possibly) nonconvex function $F\colon\R^d\to\R$ satisfying \eqn{finite-sum} where $\psi(\x)\equiv 0$ and each $f_i$ is $\ell$-smooth. Moreover, $F(\0)-F^*\leq\Delta$. The goal is to output an expected $\epsilon$-critical point $\x\in\R^d$ satisfying $\E\norm{\nabla f(\x)} \leq \epsilon$.
\end{problem}

Leveraging \citet{fang2018spider} and \citet{sidford2023quantum}, we develop a quantum algorithm that solves \prob{QFCP}:

\begin{theorem}[Informal version of \thm{FS-Q-SPIDER}]
There exist a quantum algorithm that solves \prob{QFCP} using an expected $\tilde{O}\big(n+\big(d^{1/3}n^{1/3}+\sqrt{d}\big)/\epsilon^2\big)$ queries.
\end{theorem}
Compared to~\citet{fang2018spider} using $O(n+n^{1/2}/\epsilon^2)$ queries, our quantum algorithm achieves a better query complexity when $d<\sqrt{n}$.

\paragraph{Techniques.}
In our quantum algorithms, our main contribution is leveraging quantum variance reduction with speedup on the variance reduction step in the state-of-the-art finite-sum optimization algorithms, in particular Katyusha \cite{allen2017katyusha} for the convex setting and SPIDER \cite{fang2018spider} for the nonconvex setting. Technically, our quantum speedup is originated from quantum mean estimation~\cite{montanaro2015quantum}. Classical i.i.d.~random variables with variance $\sigma^{2}$ need $\Omega(\sigma^{2}/\epsilon^{2})$ samples to approximate their mean within $\epsilon$ with high success probability~\cite{dagum2000optimal}, but quantum mean estimation by~\citet{montanaro2015quantum} achieves sample complexity $\tilde{O}(\sigma/\epsilon)$. However, most existing quantum mean estimation algorithms~\cite{hamoudi2019quantum,hamoudi2021quantum,cornelissen2022near,kothari2023mean}
have bias, which hinders their combination with Katyusha \cite{allen2017katyusha} and SPIDER \cite{fang2018spider} assuming unbiased inputs. Since Katyusha and SPIDER are iterative algorithms, bias would accumulate during the algorithm, jeopardizing their convergence guarantee.
Therefore, to achieve quantum speedup for finite-sum optimization, an unbiased quantum mean estimator is solicited. We apply the version by~\citet{sidford2023quantum}, which employs a classical multi-level Monte-Carlo (MLMC) scheme \cite{blanchet2015unbiased,asi2021stochastic} to the multivariate mean estimation algorithm by \citet{cornelissen2022near} to obtain an unbiased mean estimation.
Note that~\citet{cornelissen2023sublinear}
also developed an almost unbiased quantum mean estimation algorithm. However, it is applicable only in the one-dimensional case and retains a slight degree of bias, making it complicated when integrating with high-dimensional optimization algorithms.

In our quantum lower bounds, our primary contribution is the combination of the classical randomized lower bounds on finite-sum optimization~\cite{woodworth2016tight} with quantum adversary methods. 
Existing quantum lower bounds for optimization problems, such as~\citet{garg2020no,zhang2022quantum}, were proved by the ``zero-chain" approach. They represented quantum algorithms as sequences of unitaries and proved their lower bounds by a hybrid argument where the information accessible at each step of the algorithm is restricted. However, in the finite-sum problem, since quantum algorithms have the capability to access different sub-functions at the same time, such a hybrid argument is not valid, and we cannot assume that the algorithm accesses these sub-functions in a specific predetermined order. To address this issue, we apply the quantum adversary method first introduced in ~\citet{ambainis2000quantum}, which extends the hybrid method and takes an average over many pairs of inputs. However, the original quantum adversary method in~\citet{ambainis2000quantum,ambainis2006polynomial} were designed for lower bounding the query complexity of boolean functions, but finite-sum optimization outputs a vector rather than a single bit. Therefore, we employ a more powerful version of the adversary method
in~\citet{zhang2005power} that extends to general non-boolean functions.  
The proofs for our quantum query lower bounds consist of the following steps:

\begin{enumerate}
    \item Adapt the hard instance (\defn{case1-hardfunc}) corresponding to the random algorithms introduced by \citet{woodworth2016tight}. 
    Prove that we need to obtain enough information from over half of the sub-functions $f_i$ to find the $\epsilon$-optimal point. Notice that in the classical case, it is proved that we must obtain information about the vectors later in the sequence of a sub-function to find the $\epsilon$-optimal point. In our proof, to match with the adversary method, we establish a stronger conclusion: we need to acquire information about every vector to find the $\epsilon$-optimal point of a sub-function (\lem{sc-hardfunc-reduction}).

    \item Construct the hard instance such that each sub-function $f_i$ can only be accessed sequentially to obtain the function's construction. Reduce the problem to a quantum computing problem of determining the values of all elements in a binary matrix with a sequential verification oracle (\prob{MDP} in appendices).

    \item For the reduced problem, employ quantum adversary methods (\lem{adversary-method} in appendices) to prove its corresponding quantum complexity lower bound.
\end{enumerate}

\paragraph{Open questions.}
Our work leaves several natural directions for future investigation:
\vspace{-2mm}
\begin{itemize}
\item Can we give quantum algorithms for finite-sum optimization with better complexities? On the other hand, can we prove tighter quantum lower bounds on finite-sum optimization? Specifically, we prove our quantum lower bound by reducing to  the matrix detection problem (\prob{MDP} in appendices), which for an $n\times k$ matrix has quantum query lower bound of $\Omega(n\sqrt{k})$. It is worth investigating whether we can improve the lower bound for this problem or other hard instances.

\item Can quantum algorithms provide speedup for finding second-order stationary points in finite-sum nonconvex optimization? This had been systematically investigated in the classical setting by NEON~\cite{xu2018first,allen2018neon2} and SPIDER~\cite{fang2018spider}. Since our quantum algorithm for finding critical points (\algo{FS-Q-SPIDER}) was built upon SPIDER, it is worth investigating whether quantum speedup for finding second-order stationary points can be given.

\item Can we apply our quantum algorithm for solving machine learning problems with speedup? Several quantum algorithms for relevant machine learning problems have been proposed, for instance support vector machine~\cite{rebentrost2014quantum} and regression~\cite{wang2017quantum,liu2017fast,chen2021quantum,shao2023improved}, but the quantum speedup from variance reduction in the finite sum is widely open.
\end{itemize}


\section{Preliminaries}
\paragraph{Notation.}
We use bold letters, i.e., $\x,\y$ to denote vectors and use $\|\cdot\|$ to denote the Euclidean norm.

For a $d$-dimensional random variable $X$, we refer to the trace of the covariance matrix of $X$ as its variance, denoted by $\Var[X]$. We define $[n]\coloneqq\{1,\ldots,n\}$. By default, the logarithms are in base 2.

To model a classical probability distribution $p$ over $\R^d$ in the quantum setting, we can use the quantum state $\sum_{\x\in\R^d}\sqrt{p(\x)}\ket{\x}$. If we measure this state, the measurement outcome is described by the probability density function $p$. When applicable, we use $\ket{\mathrm{garbage}(\cdot)}$ to represent possible garbage states\footnote{The garbage state serves as a quantum counterpart to classical garbage information that emerges during the preparation of the classical stochastic gradient oracle that cannot be erased or uncomputed. In this work, we adopt a broad model without making specific assumptions about the garbage state. For a comparable discussion on this conventional usage of garbage quantum states, refer to \citet{gilyen2020distributional,sidford2023quantum} for similar discussions of this standard use of garbage quantum states.} that emerge during the implementation of a quantum oracle.

\vspace{-2mm}
\paragraph{Quantum variance reduction.}
Mean estimation a well-studied problem in quantum computing \cite{hamoudi2021quantum,cornelissen2022near,kothari2023mean}, which collectively demonstrate a quadratic quantum speedup for mean estimation. However, the output of these quantum mean estimation algorithms may exhibit bias, posing a limitation when integrating them with optimization algorithms that assume unbiased inputs. This concern was addressed in \citet{sidford2023quantum} where they developed the \textit{quantum variance reduction} algorithm that eliminates the bias leveraging the multilevel Monte Carlo (MLMC) technique. In particular, they proved the following result.

\begin{lemma}[Theorem 4 of \citet{sidford2023quantum}]\label{lem:QVR}
For a $d$-dimensional random variable $X$ with $\Var[X]\leq\sigma^2$ and some $\hat{\sigma}\geq 0$, suppose we are given access to a quantum sampling oracle
\begin{align*}
O_X\ket{0}\to\sum_{\x\in\R^d}\sqrt{\Pr(X=\x)}\otimes\ket{\mathrm{garbage}(\x)},
\end{align*}
there exists a quantum algorithm\\ \texttt{QuantumVarianceReduction}$(O_X,\hat{\sigma})$ that outputs an unbiased estimate $\hat{\mu}$ of $\E[X]$ satisfying $\E\|\hat{\mu}-\E[X]\|^2\leq\hat{\sigma}^2$ using an expected $\tilde{O}(d^{1/2}\sigma/\hat{\sigma})$ queries to $O_X$.
\end{lemma}


\section{Quantum Algorithms in Convex Settings}\label{sec:convex-algorithms}

In this section, we present our quantum algorithms for solving all the four cases of \prob{QFCO}.

\subsection{Strongly convex setting}
In this subsection, we present our quantum algorithm for Case 1 of \prob{QFCO}. Our approach is based on the framework of \texttt{Katyusha} developed in~\citet{allen2017katyusha}, which is an accelerated stochastic variance reduction algorithm. At the beginning of the algorithm, \texttt{Katyusha} first computes the exact gradient of $\x_0$ by querying $\nabla f_i(\x_0)$ for every $i\in[n]$. Then for further iterations $\x_t$ that are not very far away from $\xref\coloneqq\x_0$,  \texttt{Katyusha} uses $\nabla f(\xref)$ as a reference to approximate $\nabla f(\x_t)$ by approximating $\nabla f(\x_t)-\nabla f(\xref)$, which has a small norm when $\x_t$ is close to $\xref$ given that each $f_i$ is $\ell$-smooth. Whenever the current iteration is too far away from the reference, \texttt{Katyusha} computes the exact gradient of this iteration and makes it the new reference.

Compared to their algorithm, our algorithm (\algo{Q-Katyusha}) replaces the variance reduction step of computing $\nabla f(\x_t)-\nabla f(\xref)$ by the quantum variance reduction technique in \citet{sidford2023quantum}, as shown in \algo{QVRG}.

\begin{algorithm2e}[ht]
	\caption{Q-Katyusha}
	\label{algo:Q-Katyusha}
	\LinesNumbered
	\DontPrintSemicolon
    \KwInput{Function $F\colon\R^d\to\R$, precision $\epsilon$, smoothness $\ell$, strong convexity $\mu$}
    \KwParameter{$S=5(1+\ell^{1/2}(bm\mu)^{-1/2})\log(\Delta/\epsilon)$, $b=\lceil n^{2/3}d^{-1/3}\rceil, m=\lceil n^{2/3}d^{-1/3}\rceil$ }
    \KwOutput{An $\epsilon$-optimal point of $f$}
    $\tau_2\leftarrow \frac{1}{2b}$, $\tau_1\leftarrow\tau_2\cdot\min\big\{\sqrt{\frac{8bm\mu}{3L}},1\big\}$\;
    $\y_0=\z_0=\tilde{\x}^0\leftarrow \0$\;
    \For{$s=0,1,\ldots,S-1$}{
        $\gamma^s\leftarrow\nabla f(\tilde{\x}^s)$\label{lin:Katyusha-full-gradient}\;
        \For{$j=0,1,\ldots,m-1$}{
        $k\leftarrow (sm)+j$\;
        $\x_{k+1}\leftarrow\tau_1\z_k+\tau_2\tilde{\x}^s+(1-\tau_1-\tau_2)\y_k$\;
        $\hat{\g}_{k+1}\leftarrow\texttt{QVRG}(\x_{k+1},\tilde{\x}_s,\ell\|\x_{k+1}-\tilde{\x}_s\|/\sqrt{b})$\label{lin:Katyusha-QVRG}\;
        $\tilde{\nabla}_{k+1}\leftarrow\gamma^s+\hat{\g}_{k+1}$\;
        $\z_{k+1}\leftarrow\arg\min_{\z}\big\{\frac{3\tau_1 \ell}{2}\|\z-\z_k\|^2+\<\tilde{\nabla}_{k+1},\z\>+\psi(\z)\big\}$\;
        $\y_{k+1}\leftarrow$

        $\arg\min_{\y}\big\{\frac{3\ell}{2}\|\y-\x_{k+1}\|^2+\<\tilde{\nabla}_{k+1},\y\> +\psi(\y)\big\}$\;
        }
        $\tilde{\x}^{s+1}\leftarrow\frac{\sum_{j=0}^{m-1}\y_{sm+j+1}}{m}$\;
    }
    \Return $\x^{\mathrm{out}}\leftarrow\frac{\tau_2m\tilde{\x}^S+(1-\tau_1-\tau_2)\y_{Sm}}{\tau_2m+(1-\tau_1-\tau_2)}$\;
\end{algorithm2e}

\begin{algorithm2e}[h]
	\caption{Quantum variance-reduced gradient (QVRG)}
	\label{algo:QVRG}
	\LinesNumbered
	\DontPrintSemicolon
    \KwInput{Function $F\colon\R^d\to\R$, $\x,\xref\in\R^d$, accuracy $\hat{\sigma}$}
    Denote $\g_i\coloneqq\nabla f_i(\x)-\nabla f_i(\xref)$ for all $i\in[n]$. Implement the oracle
    $O_{\g}\ket{0}\to\frac{1}{\sqrt{n}}\sum_i\ket{\nabla f_i(\x)-\nabla f_i(\xref)}\otimes\ket{\text{garbage}(i)}.$\label{lin:QVRG-Og}\;
    $\hat{\g}\leftarrow$\texttt{QuantumVarianceReduction}($O_\g,\hat{\sigma}$)
    \Return $\hat{\g}$
\end{algorithm2e}

\begin{theorem}\label{thm:Q-Katyusha}
\algo{Q-Katyusha} solves Case 1 of \prob{QFCO} using the following number of queries in expectation:
\begin{align*}
\tilde{O}\big(n+\sqrt{d}+\sqrt{\ell/\mu}\big(n^{1/3}d^{1/3}+n^{-2/3}d^{5/6}\big)\big).
\end{align*}
\end{theorem}

\vspace{-2mm}
Prior to proving \thm{Q-Katyusha}, we first establish an upper bound on the query complexity of running the subroutine \texttt{QVRG} (\algo{QVRG}) in \lin{Katyusha-QVRG}.

\begin{lemma}\label{lem:convex-QVRG}
If every $f_i$ is $\ell$-smooth, \algo{QVRG} outputs an unbiased estimate $\hat{\g}$ of $\bar{\g}\coloneqq\frac{1}{n}\sum_{i=1}^n(\nabla f_i(\x)-\nabla f_i(\xref))$ satisfying $\E\|\hat{\g}-\bar{\g}\|\leq\hat{\sigma}^2$ using an expected $\tilde{O}(d^{1/2}\ell\|\x-\xref\|/\hat{\sigma})$ queries to $O_F$.
\end{lemma}

The proof of \lem{convex-QVRG} is deferred to \append{convex-proofs}.

The following result from \citet{allen2017katyusha} bounds the rate at which \algo{Q-Katyusha} decreases the function error of $F$. Note that correctness of this result relies solely on the fact that, at \lin{Katyusha-QVRG}, the variance of the unbiased estimate $\hat{\g}_{k+1}$ of $\frac{1}{n}\sum_i(\nabla f_i(\x_{k+1})-\nabla f_i(\tilde{\x}_s))$ is upper bounded by $\ell\|\x_{k+1}-\tilde{\x}_s\|/\sqrt{b}$, regardless of its implementation.

\begin{lemma}[Theorem 5.2,~\citet{allen2017katyusha}]\label{lem:Katyusha-convergence}
In Case 1 of \prob{QFCO}, for any $b,m\in[n]$, the output $\xout$ of \algo{Q-Katyusha} satisfies
\begin{align}\label{eqn:Katyusha-inequality}
\E[F(\xout)]-F^*\leq
\begin{cases}
O\left(\left(1+\sqrt{b\mu/(6\ell m)}\right)^{-Sm}\cdot\Delta\right),\\ \quad\qquad\quad\text{if }\frac{m\mu b}{\ell}\leq\frac{3}{8}\text{ and }b\leq m, \\
O\left(\left(1+\sqrt{\mu/(6\ell )}\right)^{-Sm}\cdot\Delta\right),\\ \quad\qquad\quad\text{if }\frac{m^2\mu}{\ell}\leq\frac{3}{8}\text{ and }b> m, \\
O\left(1.25^{-S}\cdot\Delta\right),\qquad\text{otherwise}.
\end{cases}
\end{align}
\end{lemma}

Equipped with these results, we can prove \thm{Q-Katyusha} now.

\vspace{-2mm}
\begin{proof}[Proof of \thm{Q-Katyusha}]
Given the parameter choices of \algo{Q-Katyusha}, its output $\xout$ satisfies either the first case or the third case of \eqn{Katyusha-inequality} in \lem{Katyusha-convergence}. Hence, the output $\xout$ of \algo{Q-Katyusha} satisfies
\vspace{-2mm}
\begin{align*}
&\E[F(\xout)-F^*]\leq\max\{O\Big(\left(1+\sqrt{b\mu/(6\ell m)}\right)^{-Sm}\cdot\Delta\Big),\\
&\qquad\qquad O\left(1.25^{-S}\cdot\Delta\right)\}=O(\epsilon)
\end{align*}

\vspace{-2mm}
The query complexity of \algo{Q-Katyusha} is a combination of two components: the complete gradient computation step in \lin{Katyusha-full-gradient} and the \texttt{QVRG} step in \lin{Katyusha-QVRG}.  Throughout \algo{Q-Katyusha}, there are in total $S$ full gradient computation steps and each step takes $O(n)$ queries by using $O_F$ just as a classical finite sum oracle, i.e., we query $O_F$ without employing quantum superposition. As for the second part, as per \lem{convex-QVRG}, each call to \texttt{QVRG} takes an expected
\begin{align*}
\tilde{O}\left(d^{1/2}\ell\|\x_{k+1}-\tilde{\x}_s\|/\hat{\sigma}\right)=\tilde{O}\big(\sqrt{bd}\big)
\end{align*}
queries to the quantum finite-sum oracle.
Consequently, to find an expected $\epsilon$-optimal point of $F$, the overall query complexity equals
\begin{align*}
&S\cdot O(n)+Sm\cdot\tilde{O}\big(\sqrt{bd}\big)\\
&=\tilde{O}\big(n+\sqrt{d}+\sqrt{\ell/\mu}\big(n^{1/3}d^{1/3}+n^{-2/3}d^{5/6}\big)\big).
\end{align*}
\end{proof}

\vspace{-2mm}
\subsection{Corollaries for non-smooth or non-strongly convex settings}
The algorithm for \prob{QFCO} in the non-strongly convex setting can be obtained via applying a black-box reduction introduced in~\citet{allen2016optimal}.

\begin{definition}[HOOD property,~\citet{allen2016optimal}]
We say an algorithm solving Case 1 of \prob{QFCO} satisfies the homogenous objective decrease (HOOD) property with query complexity $\mathcal{Q}(\ell,\mu)$ if for every starting point $\x_0$, it produces output $\xout$ such that
\begin{align*}
\E[F(\xout)]-F^*\leq(F(\x_0)-F^*)/4
\end{align*}
using an expected $\mathcal{Q}(\ell,\mu)$ queries.
\end{definition}

Setting $\epsilon=(F(\x_0)-F^*)/4$ in \thm{Q-Katyusha} gives:
\begin{corollary}\label{cor:Katyusha-HOOD}
\algo{Q-Katyusha} satisfies the HOOD property with $\mathcal{Q}(\ell,\mu)=\tilde{O}\big(n+\sqrt{d}+\sqrt{\ell/\mu}\big(n^{1/3}d^{1/3}+n^{-2/3}d^{5/6}\big)$.
\end{corollary}

\begin{lemma}[Theorem 3.4,~\citet{allen2017katyusha}]\label{lem:HOOD-reduction}
Given a quantum algorithm $\mathcal{A}$ that satisfies the HOOD property with query complexity $\mathcal{Q}(\ell,\mu)$, there exist three quantum algorithms that separately solves
\begin{itemize}[leftmargin=*]
\item Case 2 of \prob{QFCO} using $\sum_{s=0}^{S-1}\mathcal{Q}\big(\ell,\frac{\tilde{\mu}}{2^s}\big)$ queries, where $\tilde{\mu}=\frac{F(\0)-F^*}{\|\x^*\|^2}$ and $S=\log\frac{F(\0)-F^*}{\epsilon}$,
\item Case 3 of \prob{QFCO} using $\sum_{s=0}^{S-1}\mathcal{Q}\big(\frac{2^s}{\lambda},\mu\big)$ queries, where $\lambda=\frac{F(\0)-F^*}{L^2}$ and $S=\log\frac{F(\0)-F^*}{\epsilon}$, and
\item Case 4 of \prob{QFCO} using $\sum_{s=0}^{S-1}\mathcal{Q}\big(\frac{2^s}{\lambda},\frac{\tilde{\mu}}{2^s}\big)$ queries, where $\lambda\kern-0.5mm=\kern-0.5mm\frac{F(\0)-F^*}{L^2}$, $\tilde{\mu}\kern-0.5mm=\kern-0.5mm\frac{F(\0)-F^*}{\|\x^*\|^2}$, and $S\kern-0.5mm=\kern-0.5mm\log\frac{F(\0)-F^*}{\epsilon}$.
\end{itemize}
\end{lemma}

Combining \cor{Katyusha-HOOD} and \lem{HOOD-reduction}, we can have the following corollary.

\begin{corollary}\label{cor:HOOD-quantum}
There exists 3 quantum algorithm that solves Case 2, 3, 4 of \prob{QFCO}, respectively, with the following query complexities:
\item Case 2: $\tilde{O}\big(n+\sqrt{d}+R\sqrt{\ell/\epsilon}\big(n^{1/3}d^{1/3}+n^{-2/3}d^{5/6}\big)\big)$;
\item Case 3: $\tilde{O}\big(n+\sqrt{d}+L\big(n^{1/3}d^{1/3}+n^{-2/3}d^{5/6}\big)/\sqrt{\lambda\mu}\big)$;
\item Case 4: $\tilde{O}\big(n+\sqrt{d}+LR\big(n^{1/3}d^{1/3}+n^{-2/3}d^{5/6}\big)/\epsilon\big)$.
\end{corollary}

The proof of \cor{HOOD-quantum} is deferred to \append{convex-proofs}.


\section{Quantum Algorithm in the Nonconvex Setting}\label{sec:nonconvex-algorithm}

In this section, we present our quantum algorithm that solves \prob{QFCP}. Our approach builds upon the SPIDER algorithm \cite{fang2018spider}, which is a variance reduction technique that can estimate the gradient of an iterate with lower cost by utilizing the smoothness of each $f_i$ and reuse the gradient estimations of previous iterations. Our algorithm is a specialization of the SPIDER algorithm, where we replace the classical variance reduction step by \texttt{QVRG} introduced in \sec{convex-algorithms}.

\begin{algorithm2e}
	\caption{Finite-Sum-Q-SPIDER}
	\label{algo:FS-Q-SPIDER}
	\LinesNumbered
	\DontPrintSemicolon
    \KwInput{Function $f\colon\R^d\to\R$, precision $\epsilon$, smoothness $\ell$}
    \KwParameter{$q=\lceil n^{2/3}d^{-1/3}\rceil$, $\hat{\epsilon}=\frac{\epsilon}{5}$, total iteration budget $\mathcal{T}=\lceil\frac{4\ell\Delta}{\epsilon^2}\rceil$}
    \KwOutput{An $\epsilon$-critical point of $f$}
    Set $\x_0\leftarrow \0$\;
    \For{$t=0,1,2,\ldots,\mathcal{T}$}{
        \If{$\mod(t,q)=0$}{
            $\vect{v}_t\leftarrow\nabla f(\x_t)$\label{lin:SPIDER-full-gradient}\;
        }
        \Else{
            $\g_t\leftarrow\texttt{QVRG}(\x_t,\x_{t-1},\hat{\epsilon}/\sqrt{2q})$\label{lin:SPIDER-QVRG}\;
            $\vect{v}_t\leftarrow\vect{v}_{t-1}+\g_t$\;
        }
        \lIf{$\|\vect{v}_t\|\leq \hat{\epsilon}$}{
            \Return $\x_t$;
        }
        \lElse{
            $\x_{t+1}\leftarrow\x_t-\frac{\hat{\epsilon}}{2\ell}\cdot\frac{\vect{v}_t}{\|\vect{v}_t\|}$
        }
    }
    Uniformly randomly choose $\xout$ from $\x_0,\ldots,\x_{\mathcal{T}-1}$\;
    \Return $\xout$\;
\end{algorithm2e}

Compared to Algorithm 7 of \citet{sidford2023quantum}, which is another quantum algorithm based on SPIDER \cite{fang2018spider}, our \algo{FS-Q-SPIDER} works in the finite-sum setting which enables us to compute the full gradient in \lin{SPIDER-full-gradient}. Moreover, we carefully choose the parameters of the algorithm, taking into account of the difference in query complexity between \texttt{QVRG} and the classical variance reduction step, with a corresponding convergence analysis showing that the output of \algo{FS-Q-SPIDER} still converges to an $\epsilon$-critical point despite the changes in the parameters. In particular, we prove:

\begin{theorem}\label{thm:FS-Q-SPIDER}
\algo{FS-Q-SPIDER} solves \prob{QFCP} using the following number of queries in expectation:
\vspace{-2mm}
\begin{align*}
\tilde{O}\left(n+\ell\Delta\left(d^{1/3}n^{1/3}+\sqrt{d}\right)/\epsilon^2\right).
\end{align*}
\end{theorem}
\vspace{-3mm}
The proof of \thm{FS-Q-SPIDER} is deferred to \append{nonconvex-proofs}.


\section{Quantum Lower Bounds}
\label{sec:convex-lowerbound}
In this section, we establish our quantum lower bounds for \prob{QFCO}. Our approach leverages the framework of \citet{woodworth2016tight}, which gives a ``hard instance" based on randomly selected orthogonal spaces.
We show that for the hard function presented in this paper, a quantum algorithm cannot find an \emph{$\epsilon$-optimal} solution until it has made enough queries to the quantum oracles.

Regarding the transition from the classical hard function to the quantum complexity lower bound, our approach is based on the adversary lower bound introduced by \citet{zhang2005power} that applies to non-boolean functions.
Initially, we reduce the task of finding the \emph{$\epsilon$-optimal} point for the hard function to ``finding all elements on several chains". Subsequently, we apply the adversary method to obtain the corresponding quantum complexity lower bound.


\subsection{Smooth and strongly convex setting}
We first consider Case 1 of \prob{QFCO}.
Without loss of generality, we assume that $\ell = 1$ and $n$ is even. If $n$ is odd, we can simply take one of the sub-functions to 0, and the query complexity is reduced by a factor proportional to $\frac{n-1}{n}$. Then we define the following hard instance.
\begin{definition}[Hard instance for Case 1 of \prob{QFCO}]\label{defn:case1-hardfunc}
For constants $k$, $C$, and $\zeta$ to be decided upon later. Let $\tilde{\mu} := n\cdot\mu$. For $i = 1,2,\ldots, \lfloor n/2 \rfloor$, we define
\vspace{-2mm}
\begin{align*}
f_{i,1}(\x) &\coloneqq \frac{1-\tilde{\mu}}{16}( \inner{\x}{\v_{i,0}}^2 - 2C\inner{\x}{\v_{i,0}} \\
&\quad+ \sum_{r\text{ even}}^k \phi_c\left(\inner{\x}{\v_{i,r-1}} - \inner{\x}{\v_{i,r}}\right) );   \\
f_{i,2}(\x) &\coloneqq \frac{1-\tilde{\mu}}{16} ( \zeta\phi_c(\inner{\x}{\v_{i,k}}) \\
&\quad+\sum_{r\text{ odd}}^k \phi_c\left(\inner{\x}{\v_{i,r-1}} - \inner{\x}{\v_{i,r}}\right) ),
\end{align*}
where $\{\v_{i,r}\}$ are orthonormal vectors randomly selected from $\mathbb{R}^d$, and $\phi_c$ is a 4-smooth helper function defined as follows:
\begin{align}
    \phi_c(z) \coloneqq \begin{cases} 0 & \abs{z} \leq c; \\ 2(\abs{z} - c)^2 & c < \abs{z} \leq 2c; \\ z^2 - 2c^2 & \abs{z} > 2c \end{cases}.\label{eq:smooth-helper-func}
\end{align}

Then the hard instance is defined as follows.
\vspace{-2mm}
\begin{gather*}
    f(\x) = \frac{1}{n}\sum_{i=1}^{\frac{n}{2}}\sum_{j=1}^{2}f_{i,j}(\x),\quad\psi(\x) = \frac{\mu}{2} \norm{\x}^2.
\end{gather*}

\end{definition}

\vspace{-2mm}
This hard function has the following property: for $f_{i,1}$ and $f_{i,2}$, when we only know the values of $\v_{i,j}$ for $j\le t$, then after a new query, we can only find a vector $\x$ such that $|\<\x,\v_{i,t+1}\>|>\frac{c}{2}$. However, with large probability, $\x$ satisfies $|\<\x,\v_{i,j}\>|\leq \frac{c}{2}$ for all $j\ge t+1$. This property forms the basis for proving the lower bound, since the $\epsilon$-optimal point is related to every vector $\v_{i,j}$, as shown in the following lemma.

\begin{lemma}\label{lem:sc-hardfunc-reduction}
Denote $Q=\frac{1}{2}\big(\frac{1}{n\mu}-1\big)+ 1$, $q=\frac{\sqrt{Q}-1}{\sqrt{Q}+1}$. For  $\zeta = 1 -q $, $ k = \lfloor \frac{\sqrt{Q}-1}{4}\log \frac{\Delta}{n\epsilon(\sqrt{Q}-1)^2)}\rfloor-1$, $C = \sqrt{\frac{\Delta}{\mu}}\frac{4}{\sqrt{Q}-1}$, $c = \min{\left\{\frac{1}{\sqrt{N}}, \sqrt{\frac{8n\epsilon}{(1-\tilde\mu)(k+1)}}, Cq^{k+1}\right\}}$. For any $\epsilon \leq \frac{4}{3}\mu\Delta$ and any $\x\in\R^d$, if there exists a vector $\v_{ij}$ ($i \in \{1,2,\ldots, \lfloor n/2 \rfloor\}$, $j\in \{1,2,\ldots, k\}$) with $|\<\x,\v_{ij}\>|<\frac{c}{2}$, then $\x$ cannot be an $\epsilon$-optimal point of the hard function $F(\x)$ defined in \defn{case1-hardfunc}.
\end{lemma}

\lem{sc-hardfunc-reduction} shows that the $\epsilon$-optimal point must have a relatively large inner product with \textbf{all vectors} $\v_{ij}$. Consequently, any algorithm seeking an $\epsilon$-optimal point must find a vector with a significant inner product with all $\v_{ij}$. Then, we can reduce the problem of finding the $\epsilon$-optimal point to the problem of finding all elements of several sub-functions, where each oracle query can only reveal the next element of one sub-function with high probability. 
We can then use the adversary method to establish a lower bound on the query complexity for the quantum query problem described above. Note that \lem{sc-hardfunc-reduction} is stronger than that in~\citet{woodworth2016tight}, where it is only proved that for a specific $t$, at least half of the sub-functions $f_i$ satisfy that the $\epsilon$-optimal point must have a significant inner product with $\v_{ij_i}$ for some $j_i>t$.
\begin{remark}
    If we directly use the property proved in~\citet{woodworth2016tight}, to prove that the output $x$ of an algorithm is far away from being $\epsilon$-optimal, we must ensure that $x$ satisfies the following property: for more than half of the sub-functions with index $i$, $\abs{\inner{\x}{\v_{i,j_i}}}<\frac{c}{2}$ holds true for all $j_i>t$. To violate this property, it is sufficient to find information about only one vector that has a large inner product with $x$ for each pair of sub-functions. This case can at most correspond to an $n\times 1$ quantum query problem and thus the quantum adversary method would only yield a trivial lower bound of $n$. However, under the property given in \lem{sc-hardfunc-reduction}, we impose a higher requirement for an efficient algorithm that finds an $\epsilon$-optimal point: it must find information about all vectors in each pair of sub-functions. This naturally reduces to an $n \times k $ quantum query problem, which is evidently more difficult than the former one.
\end{remark}

Now we have transformed the lower bound of an hard instance in the optimization problem into the quantum query lower bound of the following query problem:

\begin{problem}[Multi Chain Problem] \label{prob:MCP}
    For input size $n$ and $k$, we are given oracle access to $n$
    $k$-bit strings $x_1,x_2,\cdots x_n$. Specifically,
    for a set of strings $x =[x_1,x_2,\cdots x_n]$, a query is specified by a set $(i,j,s)$, where $s$ is a $j$-bit string, and returns 1 if and only if $s$ is exactly a prefix of length $j$ for $x_i$.
    The corresponding quantum oracle which allows us to query different strings at the same time is defined as follows:
    \begin{align*}\label{eqn:OX}
    &O_x\ket{i}\otimes \ket{t}\otimes \ket{x_{i1}'x_{i2}'\cdots x_{it}'}\otimes \ket{0}
    \to\\
    &\qquad\begin{cases}
    &\ket{i}\otimes \ket{j}\otimes \ket{x_{i1}'x_{i2}'\cdots x_{it}'}\otimes\ket{0},\\
    &\text{ if $x_{ij} = x_{ij}'$ for each $j\in\{1,\cdots t\}$;} \\
    &\ket{i}\otimes \ket{j}\otimes \ket{x_{i1}'x_{i2}'\cdots x_{it}'} \otimes\ket{1},\\
    &\text{ if there exists $j\in\{1,\cdots t\}$ such that $x_{ij}\neq x_{ij}'$.} \\
    \end{cases}
    \end{align*}
The goal is to output the matrix $x$.
\end{problem}
The query lower bound for this problem is proven in \append{adversary-method}. The above analysis can be summarized into the following conclusion:
\begin{corollary}\label{cor:case1-lowerbound}
For any $\ell,\mu>0$ such that $\frac{\ell}{\mu}\geq 100n$, for any $\epsilon < \frac{4}{3}\frac{\mu\Delta}{\ell}$, $d = \tilde \Omega (\frac{\Delta \ell^2 }{\mu^2n \epsilon})$, any quantum algorithm that solves Case 1 of \prob{QFCO} with success probability at least $2/3$ must make at least the following number of queries in the worst case: \vspace{-2mm}\begin{align*}
    \Omega\left(n+n^{\frac{3}{4}} \left(\frac{\ell}{\mu}\right)^{\frac{1}{4}} \log^{\frac{1}{2}}{\left(\frac{\Delta \mu}{\epsilon \ell}\right)}\right).
\end{align*}
\end{corollary}
\vspace{-2mm}
The detailed proof is deferred to \append{case1-lb-proof}.


\subsection{Smooth and non-strongly convex setting}
In the non-strongly convex setting, our hard instance is similar to the strongly convex case, except for a regularization parameter. Without loss of generality, we assume that $\ell = R =1$ and $n$ is even, and our hard instance is constructed as follows.

\begin{definition}[Hard instance for Case 2 of \prob{QFCO}]\label{defn:case2-hardfunc}
For $i = 1,\ldots, \lfloor n/2 \rfloor$, take values $C$ and $k$ to be fixed later, define
\begin{align*}
f_{i,1}(\x) &\coloneqq \frac{1}{16}( \inner{\x}{\v_{i,0}}^2 - 2C\inner{\x}{\v_{i,0}} \\
&+ \sum_{r\text{ even}}^k \phi_c\left(\inner{\x}{\v_{i,r-1}} - \inner{\x}{\v_{i,r}}\right) );   \\
f_{i,2}(\x) &\coloneqq \frac{1}{16} (\phi_c(\inner{\x}{\v_{i,k}}) \\
&+\sum_{r\text{ odd}}^k \phi_c\left(\inner{\x}{\v_{i,r-1}} - \inner{\x}{\v_{i,r}}\right) )
\end{align*}
with orthonormal vectors $\v_{ij}$ chosen randomly on the unit sphere in $\R_d$, and $\phi_c$ is the same helper function as in \defn{case1-hardfunc}.
Then the two component of the hard instance is defined as follows:
\begin{gather*}
    f(\x) \coloneqq \frac{1}{n}\sum_{i=1}^{\frac{n}{2}}\sum_{j=1}^{2}f_{i,j}(\x),\quad\psi(\x) = 0.
\end{gather*}
\end{definition}

\begin{lemma}[]\label{lem:case2-hardfunc-reduction}
Let $k = \lfloor \frac{1}{16\sqrt{\epsilon n}}\rfloor - 1$, $C = \sqrt{\frac{6}{nk}}$, and $c = \min{\{\frac{1}{\sqrt{N}}, (2-\sqrt{3})C, 8\sqrt{\frac{\epsilon}{k}} \}}$. For any $\epsilon <\frac{1}{4096 n}$ and any $\x\in\R^d$, if for at least $\frac{n}{4}$ i's that there exists $j_i$ which holds that $j_i\leq t \coloneqq \lfloor k/2 \rfloor$ and $|\<\x,\v_{i,j_i}\>|<\frac{c}{2}$, then $\x$ cannot be an $\epsilon$-optimal point of the hard function $F$ defined in \defn{case2-hardfunc}.
\end{lemma}

Equipped with \lem{case2-hardfunc-reduction}, we prove the following result.

\begin{corollary}\label{cor:case2-lowerbound}
For any $\epsilon <\frac{\ell R^2}{4096 n}$, $d = \tilde \Omega\left(\sqrt{\frac{n}{\epsilon}}+\frac{1}{\epsilon^2}\right)$, any quantum algorithm that solves Case 2 of \prob{QFCO} with success probability at least $2/3$ must make at least the following number of queries in the worst case:  \vspace{-1mm}\begin{align*}
    \Omega\left(n+\frac{n^{\frac{3}{4}}}{\log n}\left(\frac{\ell}{\epsilon}\right)^{\frac{1}{4}} R^{\frac{1}{2}}\right).
\end{align*}
\end{corollary}
\vspace{-2mm}
The detailed proof is deferred to \append{case2-lb-proof}.

\subsection{Lipschitz and non-strongly convex setting}
In the Lipschitz and non-strongly convex setting, we also assume for simplicity $n$ is even. For the required Lipschitz property, we use a new helper function $\chi_c(z)$ based on the absolute value, which is defined as follows:
\begin{align*}
    \chi_c(z) = \max \{0,|z| - c\}.
\end{align*}
Notice that this helper function is 1-Lipschitz and it hides the information about $\x$ if the norm of $\x$ is relatively small. Then, we can define $\frac{n}{2}$ pairs of sub-functions as our hard instance:
\begin{definition}[Hard instance for Case 4 of \prob{QFCO}]\label{defn:case4-hardfunc}
For $i = 1,\ldots, \lfloor n/2 \rfloor$, take values $b,c$ and $k$ to be fixed later, let
\begin{align*}
f_{i,1}(\x) &\coloneqq \frac{1}{\sqrt{2}} |b-\<\x,\v_{i,0}\>|\\
&+\frac{1}{2\sqrt{k}} \sum_{r\text{ even}}^k \chi_c\left(\inner{\x}{\v_{i,r-1}} - \inner{\x}{\v_{i,r}}\right);\\
f_{i,2}(\x) &\coloneqq \frac{1}{2\sqrt{k}} \sum_{r\text{ odd}}^k \chi_c\left(\inner{\x}{\v_{i,r-1}} - \inner{\x}{\v_{i,r}}\right).
\end{align*}
with orthonormal vectors $\v_{ij}$ chosen randomly on the unit sphere in $\R_d$. The two component of the hard instance is defined as follows
\vspace{-2mm}
\begin{gather*}
    f(\x) = \frac{1}{n}\sum_{i=1}^{\frac{n}{2}}\sum_{j=1}^{2}f_{i,j}(\x),\quad \psi(\x) = 0.
\end{gather*}
\end{definition}
\vspace{-2mm}
Similar to the smooth case, the $\epsilon$-optimal point of the hard instance can only be found after querying a sufficient number of different vectors $\v_{i,j}$.
\begin{lemma}\label{lem:case4-hardfunc-reduction}
Let $k = \lfloor \frac{1}{10\epsilon \sqrt{n}}\rfloor$, $c = \min{\{\frac{1}{\sqrt{N}},\frac{\epsilon}{\sqrt{k}}\}}$ and $b = \sqrt{\frac{2}{n(k+1)}}$. For any $\epsilon < \frac{3}{10\sqrt{n}}$ and any $\x\in\R^d$, if for at least $\frac{n}{4}$ i's that there exists $j_i$ such that $|\<\x,\v_{i,j_i}\>|<\frac{c}{2}$, then $\x$ cannot be an $\epsilon$-optimal point of the hard function $F$ defined in \defn{case4-hardfunc}.
\end{lemma}

\begin{corollary}\label{cor:case4-lowerbound}
For any $\epsilon < \frac{3LR}{10\sqrt{n}}$, $d = \tilde \Omega\left(\frac{1}{\epsilon^3\sqrt{n}}\right)$, any quantum algorithm that solves Case 4 of \prob{QFCO} with success probability at least $2/3$ must make at least the following number of queries in the worst case:  \begin{align*}
    \Omega\left(n+ n^{\frac{3}{4}}\left(\frac{LR}{\epsilon}\right)^{\frac{1}{2}}\frac{1}{\log n}\right).
\end{align*}
\end{corollary}
\vspace{-2mm}
The detailed proof is deferred to \append{case4-lb-proof}.

\subsection{Lipschitz and strongly convex setting}
In the Lipschitz and strongly convex setting, we can construct a hard instance with corresponding quantum lower bound by adding a regularizer to \defn{case4-hardfunc}. Thus, we can directly use the result of \cor{case4-lowerbound}.Technical details can be found in the appendix.
\begin{corollary}\label{cor:case3-lowerbound}
For any $\epsilon < \frac{9L^2}{200n\mu}$ and  $d = \tilde \Omega\left(\frac{1}{\sqrt{\epsilon^3 n}}\right)$, any quantum algorithm that solves Case 3 of \prob{QFCO} with success probability at least $2/3$ must make at least the following number of queries in the worst case:  \vspace{-2mm}\begin{align*}
    \Omega\left(n+ n^{\frac{3}{4}}\left(\frac{1}{\epsilon\mu}\right)^{\frac{1}{4}} L^{\frac{1}{2}} \frac{1}{\log n}\right).
\end{align*}
\end{corollary}
\vspace{-2mm}
The detailed proof is deferred to \append{case3-lb-proof}.

\section*{Impact Statement}
This paper presents work whose goal is to advance theories of machine learning. Since this work is purely theoretical, we do not have specific societal consequences to highlight here.

\section*{Acknowledgements}
Y. Zhang and T. Li were supported by the National Natural Science Foundation of China (No. 62372006 and No. 92365117), and The Fundamental Research Funds for the Central Universities, Peking University. C. Fang was supported by the NSF National Natural Science Foundation of China (No. 62376008).


\bibliography{finite-sum}
\bibliographystyle{icml2024}

\newpage
\appendix
\onecolumn


\section{Proofs Details for Quantum Algorithms in  Convex Settings}\label{append:convex-proofs}

We give proof details for claims in \sec{convex-algorithms} here.

\begin{proof}[Proof of \lem{convex-QVRG}]
First observe that one query to $O_\g$ defined in \lin{QVRG-Og} in \algo{QVRG} can be implemented by applying $O_F$ twice to the state
\begin{align*}
\frac{1}{\sqrt{n}}\sum_i\ket{i}\otimes\ket{\x}\otimes\ket{\xref}\otimes\ket{0}\otimes\ket{0}
\end{align*}
to obtain the state
\begin{align*}
\frac{1}{\sqrt{n}}\sum_{i}\ket{i}\otimes\ket{\x}\otimes\ket{\xref}\otimes\ket{\nabla f_i(\x)}\otimes\ket{\nabla f_i(\xref)},
\end{align*}
and making the difference between the last two registers while regarding other registers as the garbage state.

Since each $f_i$ is $\ell$-smooth, we have
\begin{align*}
\|\g_i\|=\|\nabla f_i(\x)-\nabla f_i(\xref)\|\leq\ell\|\x-\xref\|,\quad\forall i\in[n].
\end{align*}
which leads to
\begin{align*}
\E_i\|\g_i-\bar{\g}\|^2\leq\E_i\|\g_i\|^2\leq\ell^2\|\x-\xref\|^2.
\end{align*}
Then by \lem{QVR}, the subroutine \texttt{QuantumVarianceReduction} outputs an unbiased estimate $\hat{\g}$ of $\bar{\g}$ with $\E\|\hat{\g}-\bar{\g}\|\leq\hat{\sigma}^2$ using an expected $\tilde{O}(d^{1/2}\ell\|\x-\xref\|/\hat{\sigma})$ queries to $O_\g$, and thus asymptotically the same number of queries to $O_F$.
\end{proof}

\begin{proof}[Proof of \cor{HOOD-quantum}]
To solve Case 2 of \prob{QFCO}, the query complexity equals
\begin{align*}
\sum_{s=0}^{S-1}\mathcal{Q}\Big(\ell,\frac{\tilde{\mu}}{2^s}\Big)
&=\sum_{s=0}^{S-1}\tilde{O}\big(n+\sqrt{d}+\sqrt{2^s\ell/\tilde{\mu}}\big(n^{1/3}d^{1/3}+n^{-2/3}d^{5/6}\big)\big)\\
&=\tilde{O}(S(n+\sqrt{d}))+\tilde{O}\big(\sqrt{\ell/\tilde{\mu}}\big(n^{1/3}d^{1/3}+n^{-2/3}d^{5/6}\big)\big)\sum_{s=0}^{S-1}2^{s/2}\\
&=\tilde{O}\big(n+\sqrt{d}+R\sqrt{\ell/\epsilon}\big(n^{1/3}d^{1/3}+n^{-2/3}d^{5/6}\big)\big).
\end{align*}
To solve Case 3 of \prob{QFCO}, the query complexity equals
\begin{align*}
\sum_{s=0}^{S-1}\mathcal{Q}\Big(\ell,\frac{\tilde{\mu}}{2^s}\Big)
&=\sum_{s=0}^{S-1}\tilde{O}\big(n+\sqrt{2^s/(\lambda\mu)}\big(n^{1/4}d^{1/4}+\sqrt{d}\big)\big)\\
&=\tilde{O}(S(n+\sqrt{d}))+\tilde{O}\big(\big(n^{1/3}d^{1/3}+n^{-2/3}d^{5/6}\big)/\sqrt{\lambda\mu}\big)\sum_{s=0}^{S-1}2^{s/2}\\
&=\tilde{O}\big(n+\sqrt{d}+L\big(n^{1/3}d^{1/3}+n^{-2/3}d^{5/6}\big)/\sqrt{\lambda\mu}\big).
\end{align*}
To solve Case 4 of \prob{QFCO}, the query complexity equals
\begin{align*}
\sum_{s=0}^{S-1}\mathcal{Q}\Big(\ell,\frac{\tilde{\mu}}{2^s}\Big)
&=\sum_{s=0}^{S-1}\tilde{O}\big(n+\sqrt{d}+\sqrt{2^{2s}/(\lambda\tilde{\mu})}\big(n^{1/3}d^{1/3}+n^{-2/3}d^{5/6}\big)\big)\\
&=\tilde{O}(S(n+\sqrt{d}))+\tilde{O}\big(\big(n^{1/3}d^{1/3}+n^{-2/3}d^{5/6}\big)/\sqrt{\lambda\mu}\big)\sum_{s=0}^{S-1}2^{s}\\
&=\tilde{O}\big(n+\sqrt{d}+LR\big(n^{1/3}d^{1/3}+n^{-2/3}d^{5/6}\big)/\epsilon\big).
\end{align*}
\end{proof}

\section{Proof Details for Quantum Algorithm in the Nonconvex Setting}\label{append:nonconvex-proofs}

In this section, we prove \thm{FS-Q-SPIDER}.
We first present a useful lemma from \citet{fang2018spider}.
\begin{lemma}[Lemma 2 \& Lemma 4, \citet{fang2018spider}]\label{lem:SPIDER-requirement}
In the setting of \prob{QFCP}, if we have
\begin{align*}
\E[\v_t]=\nabla f(\x_t)-\nabla f(\x_{t-1}),\quad\Var[\v_t]\leq\frac{\hat{\epsilon}^2}{2q}
\end{align*}
for any iteration $t\in[\mathcal{T}]$ of \algo{FS-Q-SPIDER} with $\mathrm{mod}(t,q)\neq 0$, then the following inequality holds for all $t\in[\mathcal{T}]$:
\begin{align*}
\E[f(\x_{t+1})-f(\x_t)]\leq-\frac{\hat{\epsilon}}{4\ell}\E\|\vect{v}_t\|+\frac{3\hat{\epsilon}^2}{4\ell}.
\end{align*}
\end{lemma}

Equipped with \lem{SPIDER-requirement}, we present the proof of \thm{FS-Q-SPIDER} below.

\begin{proof}[Proof of \thm{FS-Q-SPIDER}]
By \lem{convex-QVRG}, for any iteration $t\in[\mathcal{T}]$ of \algo{FS-Q-SPIDER} with $\mathrm{mod}(t,q)\neq 0$ we have
\begin{align*}
\E[\g_t]=\nabla f(\x_t)-\nabla f(\x_{t-1}),\quad\Var[\g_t]\leq\frac{\hat{\epsilon}^2}{2q}.
\end{align*}
Hence, by telescoping the result from \lem{SPIDER-requirement}, we have
\begin{align*}
\frac{\hat{\epsilon}}{4\ell}\sum_{t=0}^{\mathcal{T}-1}\E\|\vect{v}_t\|\leq f(\0)-\E f(\x_{\mathcal{T}})+\frac{3\mathcal{T}\hat{\epsilon}^2}{4\ell}\leq\Delta+\frac{3\mathcal{T}\hat{\epsilon}^2}{4\ell}
\end{align*}
and
\begin{align*}
\frac{1}{\mathcal{T}}\sum_{t=0}^{\mathcal{T}-1}\E\|\vect{v}_t\|\leq\frac{4\ell\Delta}{\hat{\epsilon}\mathcal{T}}+3\hat{\epsilon}\leq 4\hat{\epsilon},
\end{align*}
where for each $t\in[\mathcal{T}]$ we have
\begin{align*}
\E\|\vect{v}_t\|=\E\|(\vect{v}_t-\nabla f(\x_t))+\nabla f(\x_t)\|\geq\E\|\nabla f(\x_t)\|-\E\|\vect{v}_t-\nabla f(\x_t)\|\geq \E\|\nabla f(\x_t)\|-\hat{\epsilon}
\end{align*}
by \lem{SPIDER-requirement}, which leads to
\begin{align*}
\E\|\xout\|=\frac{1}{\mathcal{T}}\sum_{t=0}^{\mathcal{T}-1}\E\|\nabla f(\x_t)\|+\hat{\epsilon}\leq 5\hat{\epsilon}=\epsilon,
\end{align*}
indicating that the output of \algo{FS-Q-SPIDER} is an expected $\epsilon$-critical point.

The query complexity of \algo{FS-Q-SPIDER} is a combination of two components: the complete gradient computation step in \lin{SPIDER-full-gradient} and the \texttt{QVRG} step in \lin{SPIDER-QVRG}, where each full gradient computation steps and each step takes $O(n)$ queries by using $O_F$ just as a classical finite sum oracle, i.e., we query $O_F$ without employing quantum superposition. As for the second part, as per \lem{convex-QVRG}, each call to \texttt{QVRG} takes an expected
\begin{align*}
\tilde{O}\left(\frac{d^{1/2}\ell\|\x_{t}-\x_{t-1}\|}{\hat{\epsilon}/\sqrt{2q}}\right)=\tilde{O}\big(\sqrt{dq/2}\big)
\end{align*}
queries to the quantum finite-sum oracle. Hence, for every $q$ iterations, the query complexity equals
\begin{align*}
n+q\tilde{O}(\sqrt{dq/2}),
\end{align*}
and the overall query complexity of \algo{FS-Q-SPIDER} equals
\begin{align*}
\left(1+\frac{\mathcal{T}}{q}\right)\cdot\left(n+q\tilde{O}(\sqrt{dq/2})\right)=\tilde{O}\left(n+\frac{\ell\Delta}{\epsilon^2}\left(d^{1/3}n^{1/3}+\sqrt{d}\right)\right).
\end{align*}
\end{proof}


\section{Proof Details for Quantum Lower Bounds}
In this section, we list several lemmas to prove our quantum lower bounds for convex settings of \prob{QFCO}.


\subsection{The strong weighted adversary method}
\label{append:adversary-method}
The last step of our proof is using the non-negative quantum adversary method.

\begin{lemma}[Lemma 6, ~\citet{cleve2012reconstructing}]\label{lem:adversary-method}
Let $f$ be a function from a finite set $S$ to another finite set $T$, and let $Q$ be a finite set of possible query strings. Given an unknown input $x\in S$, the oracle $O_x$ corresponding to $x$ is the unitary transformation $O_x \ket{q}\ket{a}\ket{z} =\ket{q} \ket{a\oplus \xi(x;q)} \ket{z}$, where $q$ is a query string from $Q$, $a\in\{0,1\}$ is the register of binary answer, $z$ is the auxiliary register, and $\xi : S\times Q\to \{0,1\}$ is a function that defines the response to oracle queries. Also, let $w$, $w'$ denote a weight scheme as follows:
\begin{itemize}
    \item Every pair $(x, y) \in S \times S$ is assigned a non-negative weight $w(x, y) = w(y, x)$ that satisfies $w(x, y) = 0$ whenever $f (x) = f (y)$;
    \item Every triple $(x,y,q)\in S \times S \times Q$ is assigned a non-negative weight $w'(x,y,q)$ that satisfies $w'(x,y,q)=0$ for all $x,y,q$ such that $\xi(x;q) = \xi(y;q)$ or $f(x) = f(y)$, and $w'(x,y,q)\cdot w'(y,x,q) \geq w(x,y)^2$ for all $x,y,q$ such that $\xi(x;q) \neq \xi(y;q)$ and $f(x) \neq f(y)$.
\end{itemize}
For all $x\in S$ and $q \in Q$, let $\mu(x) = \sum_y w(x,y) $ and $\nu(x,q) = \sum_y w'(x,y,q)$. Then any quantum algorithm that computes $f (x)$ with success probability at least $\frac{2}{3}$ on an arbitrary input x must make
\begin{equation*}
    \Omega \left( \min_{\substack{x,y,q;\,w(x,y)>0, \\ \xi(x,q) \neq \xi(y,q)}}  \sqrt{\frac{\mu(x)\cdot \mu(y)}{\nu(x,q)\cdot \nu(y,q)}}  \right)
\end{equation*}
queries to the oracle $O_x$.
\end{lemma}

Using the strong weighted adversary method, we prove the following quantum lower bound for the multi chain problem:

\begin{lemma}\label{lem:LBofMCP}
    Any quantum algorithms that solves \prob{MCP} on $n \times k$ ($n$ strings, $k$ bits for each string) with success probability at least $\frac{2}{3}$ must take $\Omega(n\sqrt{k})$ queries.
\end{lemma}

\begin{proof}
    Our proof is similar to \citet{ambainis2014quantum}.
    In our multi-chain problem, the input is a string $x \in \{0,1\}^{n\times k }$, and $f(x)$ is defined as $f(x) = x$. Queries can be formalized as follows: $q = (i,t,x_{i1}',x_{i2}',\dots x_{it}')$, where $i\in \{1,2,\dots n\}$, $t\in \{1,2,\dots k\}$ and $x_{ij}'\in\{0,1\}$. Here $\xi(x,q) = 1$ if and only if $x_{ij}' = x_{ij}$ for every $j\in{1,2,\dots t}$. Let $d(x,y)$ denote the Hamming distance of two inputs $x$ and $y$ (i.e., the number of differing bits between the two inputs.). Then we define the following weight schemes:
    \begin{itemize}
        \item $w(x,y) = 1$ if $d(x,y) = 1$, and $w(x,y) = 0$ otherwise;
        \item $w'(x,y,q) = w'(y,x,q) = 1$ if $d(x,y) = 1$ and $\xi(x,q) \neq \xi(y,q)$, and $w'(x,y,q) = w'(y,x,q) = 0$ otherwise.
    \end{itemize}

    Then for any $x\in \{0,1\}^{n\times k }$, we have $\mu(x) = nk $ and
    \begin{align*}
        \nu(x,q) = \sum_{y} \textbf{1}_{\{y:~d(x,y) = 1 ,~ \xi(x,q) \neq \xi(y,q)\} }(y) =
        \begin{cases}
            t,\quad \text{if } \xi(x,q) = 1;\\
            1,\quad \text{if } \xi(x,q) = 0 \text{ and } |j: j\leq t\,~x_{ij} \neq x_{ij}'| = 1;\\
            0, \quad \text{otherwise.}
        \end{cases}
    \end{align*}
    Therefore,
    \begin{equation*}
        \min_{\substack{x,y,q;\,w(x,y)>0, \\ \xi(x,q) \neq \xi(y,q)}}  \sqrt{\frac{\mu(x)\cdot \mu(y)}{\nu(x,q)\cdot \nu(y,q)}} = \min \sqrt{\frac{nk\cdot nk}{1\cdot t}} = \sqrt{\frac{nk\cdot nk}{1\cdot k}} = n\sqrt{k}.
    \end{equation*}

    By \lem{adversary-method}, any quantum algorithms must take $\Omega(n\sqrt{k})$ queries to solve the multi chain problem with success probability more than $\frac{2}{3}$.
\end{proof}


\subsection{Reduction from optimization to \prob{MCP}}

\begin{problem}[Matrix Detection Problem]\label{prob:MDP}
Given a matrix where each element is a vector, and each vector can take one of two known possible values. Specifically, given a $n\times k$ matrix $A$, each element $\vect{a}_{ij}$ can be one of the vectors $\v_{ij0}$ or $\v_{ij1}$. We have the following quantum oracle $O_{A}$:
\begin{align*}
O_A\ket{i}\otimes \ket{j}\otimes \ket{m_1m_2\cdots m_j}\otimes\ket{0}\ket{0}
\to\begin{cases}
    \ket{i}\otimes \ket{j}\otimes \ket{m_1m_2\cdots m_j}\otimes\ket{1}\ket{0}\\
    \qquad\text{ if $\vect{a}_{i,p}= \v_{i,p,m_p}$ for each $p\in\{1,2,,\cdots j\}$ and } \vect{a}_{i,j+1} = \v_{i,j+1,0};\\
    \ket{i}\otimes \ket{j}\otimes \ket{m_1m_2\cdots m_j}\otimes\ket{1}\ket{1}\\
    \qquad\text{ if $\vect{a}_{i,p} = \v_{i,p,m_p}$ for each $p\in\{1,2,\cdots j\}$ and } \vect{a}_{i,j+1} = \v_{i,j+1,1};\\
    \ket{i}\otimes \ket{j}\otimes \ket{m_1m_2\cdots m_j}\otimes\ket{0}\ket{0},\\
    \qquad\text{ if there exists $p\in\{1,2,\cdots j\},~\vect{a}_{i,p}\neq \v_{i,p,m_p}$  }.\\
\end{cases}
\end{align*}
The goal is to output the matrix $A$.
\end{problem}
This problem is quite similar to finding all vectors in our hard function. Intuitively, for each row, we have to start from the beginning and detect the specific values of each vector in order.

\begin{lemma}\label{lem:QFCO-MDP-reduction}
    For dimension $d = \Omega\left( n(k+1) + \frac{8R^2}{c^2} \log{(2nkN^3)}\right)$, given a quantum algorithm that finds all $\v_{ij}$ of the hard instance in \defn{case1-hardfunc} with success probability more than $\frac{2}{3}$ using $s\leq N$ queries, we can construct a quantum algorithm solving  \prob{MDP} with success probability more than $\frac{2}{3}$ using $O(s)$ queries. 
\end{lemma}

\begin{proof}
First, for the hard instance $F$, we construct the following quantum oracle:
\begin{align}
    O_{\v} \ket{\x_1}\ket{\x_2}\cdots \ket{\x_k} \otimes \ket{i} \otimes \ket{0}\ket{0}\cdots \ket{0}\to
    \ket{\x_1}\ket{\x_2}\cdots \ket{\x_k} \otimes \ket{i} \otimes \ket{\x_1'}\ket{\x_2'}\cdots \ket{\x_k'}
    \label{eqn:vector-oracle}
\end{align}
where
\begin{align*}
    \x_j' =
    \begin{cases}
    \v_{i,j},\quad \text{if } \abs{\<\x_j',\v_{i,j}\>} > \frac{c}{2}\text{ or } \abs{\<\x_j',\v_{i,j-1}\>} > \frac{c}{2}\text{ or } \abs{\<\x_j',\v_{i,j+1}\>} > \frac{c}{2};\\
        \0, \quad \text{otherwise}.\\
    \end{cases}
\end{align*}

We demonstrate that when the task is to find a vector $\x$ with a large inner product with each vector $\v_{i,j}$, this oracle is stronger than the quantum finite-sum oracle \eqn{Of}. Specifically, we can simulate a query to \eqn{Of} using a query to \eqn{vector-oracle}. For example, consider the hard instance in \defn{case1-hardfunc}. For a query to oracle \eqn{Of} with input $\ket{\x} \otimes \ket{i}$, consider a query to $O_{\v}$ with input $\ket{\x}\ket{\x}\cdots \ket{\x} \otimes \ket{i}$. Observe that for $j$ with  $\max \{\<\x,\v_{i,j-1}\>,\<\x,\v_{i,j}\>,\<\x,\v_{i,j+1}\>\} \leq \frac{c}{2}$, since the helper function takes value 0, $\nabla f_i(\x)$ does not contain any information about $\v_{i,j}$. Therefore, $\nabla f_i(\x)$ can be computed by the output of the oracle $O_{\v}$.

For a quantum algorithm $A$ with access to $O_{\v}$ that can find the values of all events, suppose $A$ makes at most $N$ queries. We define ``bad event" to be
\begin{align*}
    B_{i,j,t} \coloneqq \left[ \text{the input of $t$-th query $\x_j^{(t)}$ satisfies $\abs{\<\x_j^{(t)},\v_{i,j}\>} > \frac{c}{2}$, but $\v_{i,j}$ is not the output of any previous query}\right].
\end{align*}
Intuitively, when $B_{i,j,t}$ happens, before $A$ queries $\x_j^{(t)}$,  $A$ knows at most $t-1$ vectors that have relatively small inner products with $\v_{i,j}$.

Since $\v_{i,j}$ is randomly chosen as a unit vector orthogonal to other $\v$, we can bound the probability that the inner product of a fixed unit vector and a uniformly random unit vector in $d-n(k+1)+1$ dimensions is larger than $\frac{c}{2}$. This probability can be explained as the ratio of the combined areas of the upper and lower caps with a radius of $r \coloneqq \sqrt{1-\left(\frac{c}{2}\right)^2}$ to the surface area of the sphere. Thus,
\begin{align*}
    \Pr[\abs{\<\x ,\v_{i,j}\>}\ > \frac{c}{2}] &\leq \frac{r^{d-n(k+1)+1}}{1^{d-n(k+1)+1}} = \left(1-\left(\frac{c}{2}\right)^2\right)^{\frac{d-n(k+1)+1}{2}} \leq e^{-\frac{c^2(d-n(k+1)+1)}{8}}.
\end{align*}

Therefore, the probability that the event $B_{i,j,t}$ occurs is less than the probability of finding $\x$ such that $\abs{\<\x ,\v_{i,j}\>}\ > \frac{c}{2}$ using less than $t$ queries, for some constant $h$:
\begin{align*}
    \Pr[B_{i,j,t}]\leq t^2h\cdot e^{-\frac{c^2(d-n(k+1)+1)}{8}}.
\end{align*}
 Define ``good event" $G \coloneqq \cup_{i=1}^{n/2} \cup_{j=1}^{k} \cup_{t=1}^{N} \overline{B_{i,j,t}}$. Thus, the probability that all ``bad" event do not happen is
\begin{align*}
    \Pr[G] \geq 1 - \sum_{i=1}^{n/2} \sum_{j=1}^{k} \sum_{t=1}^{N} \Pr[B_{i,j,t}] \leq 1 - \frac{1}{2}nkhN^3 e^{-\frac{c^2(d-n(k+1)+1)}{8}}.
\end{align*}

Take $d \geq n(k+1) + \frac{8}{c^2} \log{(2nkhN^3)}$, we have
\begin{align*}
    \Pr[G] \geq 1-\frac{1}{16} = \frac{15}{16}.
\end{align*}
The above proof assumes that all queries are within a sphere of radius $1$. For a general radius $R$, we can replace $c$ with $\frac{c}{R}$ to adjust the parameters accordingly.

The occurrence of the ``good" event implies that with high probability, the algorithm cannot guess the value of the vector but can only obtain information about the vector through the provided oracle $O_\v$. In such circumstance, we can simulate $O_\v$ using an another oracle $O_\v'$ that demands more precise inputs:
\begin{align}
    O_v \ket{\x_1}\ket{\x_2}\cdots \ket{\x_k} \otimes \ket{i} \otimes \ket{0}\ket{0}\cdots \ket{0}\to
    \ket{\x_1}\ket{\x_2}\cdots \ket{\x_k} \otimes \ket{i} \otimes \ket{\x_1'}\ket{\x_2'}\cdots \ket{\x_k'}
    \label{eqn:vector-oracle'}
\end{align}
where
\begin{align*}
    \x_j' =
    \begin{cases}
    \v_{i,j},\quad \text{if } \x_{j'} = \v_{i,j'} \text{ for all } j'< j;\\
        \0, \quad \text{otherwise}.\\
    \end{cases}
\end{align*}

Note that when $G$ occurs, when the algorithm queries $\x_j^{(t)}$ such that $\abs{\<\x_j^{(t)},\v_{i,j}\>} > \frac{c}{2}$, $\v_{i,j}$ has been the output of a previous query, so we can directly substitute the input with $\v_{i,j}$, and the algorithm still works using the oracle $O_\v'$. If event $B_{i,j,t}$ occurs, we simply mark it as a failure of the algorithm.

Now, we can consider that a quantum algorithm using oracle $O_F$ to find all $\v_{i,j}$'s can be implemented with a quantum algorithm using $O_\v'$ (with a small failure probability less than $\frac{1}{16}$) with the same number of queries. Adding a constraint on $\v_{i,j}$, specifically that $\v_{i,j}$ is one of the two known vectors $\v_{i,j,0}$ and $\v_{i,j,1}$, implies that the task of using $O_\v'$ to find all vectors is exactly equivalent to \prob{MDP}. Therefore, from a quantum algorithm using oracle $O_F$ that finds all $\v_{i,j}$ of the hard instance with success probability more than $\frac{2}{3}$ using $s$ queries, we can construct a quantum algorithm solving \prob{MDP} using $O(s)$ quantum queries through the aforementioned reduction process. The failure probability could be controlled by repeating the algorithm a constant number of times.
\end{proof}

Next, we present the following lemma, demonstrating that the complexity of a quantum algorithm to find vectors for all sub-functions is only different from finding vectors for half of the sub-functions by a logarithmic factor. This lemma is essential in the proof of cases 2-4 of \prob{QFCO}, as in these situations, finding an $\epsilon$-optimal point for the total function only requires acquiring information about half of the sub-functions.

\begin{lemma}\label{lem:QFCO-MDP-reduction'}
    For dimension  $d = \Omega\left( n(k+1) + \frac{8}{c^2} \log{(2nkN^3)}\right)$, given a quantum algorithm which finds vectors $\v_{i,j}$ for all $j$ of $\frac{n}{4}$ different (and uncertain) $i$'s of a hard instance with success probability more than $\frac{2}{3}$ using $s$ queries, we can construct a quantum algorithm solving \prob{MDP} with success probability more than $\frac{2}{3}$ using $O\big((s+\frac{n}{2})\log{n}\big)$ quantum queries, where $N$ is the maximum number of query times. 
\end{lemma}

\begin{proof}
    The proof is similar to the proof for \lem{QFCO-MDP-reduction}. For the task of finding a vector $\x$ such that $\abs{\<\x,\v_{i,j}\>} > \frac{c}{2}$ for all $j$ of $\frac{n}{4}$ different (and uncertain) $i$'s, $O_\v$ \eqn{vector-oracle} is stronger than the quantum finite-sum oracle $O_F$ \eqn{Of}. Specifically, we can simulate a query to \eqn{Of} using a query to \eqn{vector-oracle}. Therefore, given a quantum algorithm using oracle $O_F$, we can construct another algorithm $A$ using oracle $O_\v$ with the same number of queries. Now we directly use $A$ to construct an algorithm $B$ that finds $\x'$ such that $\<\x,\v_{i,j}\> < \frac{c}{2}$ for \textbf{all $i,j$'s} of the same hard instance.

    The algorithm $B$ works as follows: Repeat the following steps $T$ times, where $T$ will be specified later. In the $t$-th round, $B$ randomly permute the $\frac{n}{2}$ items $(f_{i,1},f_{i,2})$ for $i =1,2,\cdots \frac{n}{2}$. Then it simulates $A$ and outputs $\x_t$. After receiving $\x_t$, it queries $O_\v$ with inputs  $\ket{\x_t}\ket{\x_t}\cdots \ket{\x_t} \otimes \ket{i}$ for all $i$. Since $\x_t$ satisfies that $\abs{\<\x_t,\v_{i,j}\>} > \frac{c}{2}$ for all $j$ of $\frac{n}{4}$ different $i$'s, the above $\frac{n}{2}$ queries will receive values of $\v_{i,j}$ for all $j$ corresponding to at least $\frac{n}{4}$ different $i$'s. After the $T$ rounds, $B$ output $\sum_{i=1}^{n/2} \sum_{j=1}^{k} \v_{i,j}$ if it receives all $\v_{i,j}$, and fails otherwise.

    Now we bound the failure probability. In each round, the probability that vectors of $i$-th sub-function are found is larger than $\frac{1}{2}$, and due to the independence between each round, $\Pr[\text{ $B$ fails to find the information of $\v_{i,j}$ for all $j\leq k$ }] \leq \big(\frac{1}{2}\big)^{T} $. Hence
    \begin{align*}
        &\Pr[\text{ $B$ finds the value of $\v_{i,j}$ for all $i\leq \frac{n}{2}$ and $j\leq k$ }]
        \\ &\qquad\geq 1 - \sum_{i=1}^{\frac{n}{2}} \Pr[\text{ $B$ fails to find the value of $\v_{i,j}$ for all $j\leq k$ }]
        \\ &\qquad\geq 1 - \frac{n}{2}\cdot  \left(\frac{1}{2}\right)^{T}.
    \end{align*}
    Take $T = 4 + \lceil \log{(\frac{n}{2})}\rceil$, the success probability of $B$ is larger than $\frac{15}{16}$. The number of queries to the oracle $O_\v$ is $\Theta(\big((s+\frac{n}{2})\log{n}\big))$.

    The remaining steps are consistent with \lem{QFCO-MDP-reduction}. We can construct a quantum algorithm to solve \prob{MDP} with complexity $O(\big((s+\frac{n}{2})\log{n}\big))$. Here we repeat the algorithm a constant number of times to control the failure probability.
\end{proof}

Next, we prove that \prob{MDP} can be reduced to \prob{MCP}, for which problem we can construct a lower bound using the adversary method.

\begin{lemma}
    \label{lem:MDP-MCP-reduction}
    Given an algorithm that solves \prob{MDP} for input size $ n\times k$ with success probability larger than $\frac{2}{3}$ using $s$ queries, we can construct a quantum algorithm solving \prob{MCP} for the same input size with success probability larger than $\frac{2}{3}$ using $O(s)$ queries.
\end{lemma}

\begin{proof}
    The proof is based on a reduction between the two quantum oracles. While using existing results on adversary bounds, it is worth noting that a quantum oracle in \prob{MDP} can be implemented using two oracles from \prob{MCP}: We can simulate the results of $O_A\ket{i}\otimes \ket{j}\otimes \ket{m_1m_2\cdots m_j}\otimes\ket{0}\ket{0}$ by querying oracle $O_C$ with inputs $\ket{i}\otimes \ket{j+1}\otimes \ket{m_1m_2\cdots m_j0}\otimes\ket{0}$ and $\ket{i}\ket{j+1}\otimes \ket{m_1m_2\cdots m_j1}\otimes\ket{0}$, adding several quantum unitaries. Specifically, the first qubit of the oracle $O_A$ is $\ket{1}$ if and only if the output of two queries to the oracle $O_C$ contains at least one $\ket{1}$, which means that $(m_1,\cdots m_j) = (a_{i,1},\cdots a_{i,j})$. Under the above conditions the second qubit of the oracle $O_A$ is $\ket{0}$ when the first query of $O_C$ is $\ket{1}$ and $\ket{1}$ otherwise.

    Since the output of \prob{MDP} can be represented as a binary matrix, which is the same as \prob{MCP}, we can build an instance of \prob{MDP} out of an instance of \prob{MCP} by adding several random vectors. Therefore, We can construct an algorithm with query complexity $O(s)$ to solve \prob{MCP}.
\end{proof}


\subsection{Proofs for the smooth and strongly convex setting}\label{append:case1-lb-proof}
The proof of \lem{sc-hardfunc-reduction} is inspired by Theorem 8 of \citet{woodworth2016tight}. Note that compared to their proof, we show that to find an $\epsilon$-optimal point we need the information of all vectors rather than just one vector located later in each sub-function.

\begin{proof}[Proof of \lem{sc-hardfunc-reduction}]
    Intuitively, our proof strategy is roughly as follows: We initially decompose $F(\x)$ into the average of several sub-functions defined on orthogonal subspaces. If the projection of $\x$ onto a certain subspace has a small enough inner product with one of the randomly chosen vectors, we directly prove that the function value corresponding to $\x$ on that subspace has a sufficiently large gap from the minimum value of that sub-function.

    Firstly, in order to bound the influence of $\<\x,\v_{i,j}\>$ on the total function $F(\x)$, it is convenient to bundle together all terms affecting the value of $\<\x,\v_{i,j}\>$ from the components of all sub-functions. These terms are contained in $f_{i,1}(\x)$, $f_{i,2}(\x)$ and $\psi(\x)$. Consider the projection operator $P_i$ which projects the vector $\x$ onto the subspace spanned by the vector set $\{\v_{i,j}\}_{j=0}^k$, and define $P_{\perp}$ as the operator projecting $\x$ onto the subspace orthogonal to $\v_{i,j}$ for all $i,j$. Then we have
    \begin{align*}
        \psi(\x) = \frac{\mu}{2}\norm{x}^2 = \frac{\mu}{2} \sum_{i=1}^{n/2} \left( \norm{P_i \x}^2 \right) + \frac{\mu}{2} \norm{P_{\perp}\x}^2.
    \end{align*}
    Split $\psi(\x)$ amongst $f_{i,1}$ and $f_{i,2}$, we obtain the modified sub-functions:
    \begin{align*}
        \tilde{f}_{i,1}(\x) = f_{i,1} + \frac{\tilde \mu}{4} \norm{P_i\x}^2,\quad
        \tilde{f}_{i,2}(\x) = f_{i,2} + \frac{\tilde \mu}{4} \norm{P_i\x}^2.
    \end{align*}
    Then, we can represent the total function $F(\x)$ as follows:
    \begin{align*}
        F(\x) = \frac{1}{n}\sum_{i=1}^{n/2} \left(\tilde f_{i,1}(\x) + \tilde f_{i,2}(\x) \right) +\frac{\mu}{2} \norm{P_{\perp}\x}^2.
    \end{align*}
    Notice that there is still a remaining term $\frac{\mu}{2} \norm{P_{\perp}\x}^2$ here, but this part is not crucial for our analysis, as when minimizing the total function $F(\x)$, we can always set $P_{\perp}\x = \0$. Next, we consider the summation
    \begin{align*}
        &~~~~\frac{1}{2}\left(\tilde{f}_{i,1}(\x) + \tilde{f}_{i,2}(\x)\right) \\
        &= \frac{1-\tilde{\mu}}{32}\big( \inner{\x}{\v_{i,0}}^2 - 2C\inner{\x}{\v_{i,0}} +\zeta\phi_c(\inner{\x}{\v_{i,k}}) + \sum_{j=1}^k \phi_c\left(\inner{\x}{\v_{i,r-1}} - \inner{\x}{\v_{i,r}}\right) \big)+ \frac{\tilde \mu}{4} \norm{P_i\x}^2.
    \end{align*}

    We define
    \begin{align*}
         F_i(\x)
        &= \frac{1-\tilde{\mu}}{32}\bigg( \inner{\x}{\v_{i,0}}^2 - 2C\inner{\x}{\v_{i,0}} +\zeta\inner{\x}{\v_{i,k}}^2
+ \sum_{j=1}^k \left(\inner{\x}{\v_{i,r-1}} - \inner{\x}{\v_{i,r}}\right)^2 \bigg) + \frac{\tilde \mu}{4} \norm{P_i\x}^2
    \end{align*}
    and
    \begin{align*}
         F_i^t(\x)
        &= \frac{1-\tilde{\mu}}{32}\bigg( \inner{\x}{\v_{i,0}}^2 - 2C\inner{\x}{\v_{i,0}} +\inner{\x}{\v_{i,t}}^2
+ \sum_{j=1}^t \left(\inner{\x}{\v_{i,r-1}} - \inner{\x}{\v_{i,r}}\right)^2 \bigg)  + \frac{\tilde \mu}{4} \norm{P_i\x}^2.
    \end{align*}
    Intuitively, $F_i(\x)$ is an approximation to $\frac{1}{2}\left(\tilde f_{i,1}(\x) + \tilde f_{i,2}(\x)\right)$, and $F_i^t(\x)$ is a truncation of $F_i(\x)$ to $\v_{i,t}$, neglecting the vectors beyond  $\v_{i,t+1}$. Consider the helper function \eq{smooth-helper-func}, we know that when $|z|\leq c$, the function is constant at 0, and for any $z$,
    \begin{equation}
        z^2-2c^2 \leq \phi_c(z) \leq z^2.
        \label{eq:helper-function-property}
    \end{equation}
    According to the properties of the helper function, we have
    \begin{gather*}
        F_i(\x) \leq \frac{1}{2}\left(\tilde f_{i,1}(\x) + \tilde f_{i,2}(\x)\right) + \frac{(1-\tilde \mu)(k+\zeta)}{16}c^2\\
        \frac{1}{2}\left(\tilde f_{i,1}(\x) + \tilde f_{i,2}(\x)\right) \leq F_i(\x).
    \end{gather*}
    for any $\x$.
    For convenience, let $Q \coloneqq \frac{1}{2}(\frac{1}{\tilde \mu}-1)+1$, then
    \begin{align*}
         F_i(\x)
        &= \frac{1}{2}\bigg(\frac{\tilde{\mu}(Q-1)}{8}\Big( \inner{\x}{\v_{i,0}}^2 - 2C\inner{\x}{\v_{i,0}} +\zeta\inner{\x}{\v_{i,k}}^2
        + \sum_{j=1}^k \left(\inner{\x}{\v_{i,r-1}} - \inner{\x}{\v_{i,r}}\right)^2 \Big) + \frac{\tilde \mu}{2} \norm{P_i\x}^2\bigg),\\
        F_i^t(\x)
        &= \frac{1}{2}\bigg(\frac{\tilde{\mu}(Q-1)}{8}\Big( \inner{\x}{\v_{i,0}}^2 - 2C\inner{\x}{\v_{i,0}} +\inner{\x}{\v_{i,t}}^2
        + \sum_{j=1}^t \left(\inner{\x}{\v_{i,r-1}} - \inner{\x}{\v_{i,r}}\right)^2 \Big) + \frac{\tilde \mu}{2} \norm{P_i\x}^2\bigg).
    \end{align*}

    Let $\hat{\x}\coloneqq \arg\min_{\x}F_i(\x)$, by the first-order condition for the minimum value of $F_i(\x)$, $\hat{\x}$ satisfies
    \begin{gather*}
        2\frac{Q+1}{Q-1} \<\hat{\x},\v_{i,0}\> - \<\hat{\x},\v_{i,1}\> = C,\\
        \<\hat{\x},\v_{i,j-1}\> -2\frac{Q+1}{Q-1}\<\hat{\x},\v_{i,j}\> +\<\hat{\x},\v_{i,j+1}\> = 0,\\
        \big(1+\zeta +\frac{4}{Q-1} \big) \<\hat{\x},\v_{i,k}\> - \<\hat{\x},\v_{i,k-1}\> = 0
    \end{gather*}
    for $j = 1,2,\cdots k-1$.

    Define $q \coloneqq \frac{\sqrt{Q}-1}{\sqrt{Q}+1}$ $(q<1)$ and set $\zeta = 1 - q$. Then $\hat \x$ can be expressed as follows:
    \begin{align*}
        \hat{\x} = C \sum_{j=0}^{k} q^{j+1} \v_j
    \end{align*}
    and
    \begin{align*}
        F_i(\hat{\x}) = -\frac{\tilde{\mu} C^2}{16} (\sqrt{Q} - 1)^2.
    \end{align*}
    Therefore, the suboptimality of point $\0$ is $\epsilon_i \coloneqq F_i(\0) - F_i(\hat{\x}) = \frac{\mu C^2}{16} (\sqrt{Q} - 1)^2$.

    Now consider an arbitrary vector $\x$. If there exists an index $t$ that $\abs{\<\x,\v_{it}\>}<\frac{c}{2}$, take $c\leq Cq^{k+1}$ since $F_i$ is a $\frac{\tilde\mu}{2}$-strongly convex function, $F_i(\x) - F(\hat{\x}) \geq \frac{\tilde\mu}{4}\norm{\x - \hat{\x}}^2$. Thus
    \begin{align*}
        &\quad\frac{F_i(\x) - F(\hat{\x})}{F_i(\0) - F(\hat{\x})}\\
        &\geq \frac{\frac{\tilde\mu}{4}\norm{\x - \hat{\x}}^2}{\frac{\mu C^2}{16} (\sqrt{Q} - 1)^2}\\
        &\geq \frac{4}{C^2}\cdot \frac{(\abs{C q^{k+1}} - \abs{c})^2}{(\sqrt{Q}-1)^2}\\
        &\geq \frac{1}{C^2}\cdot \frac{C^2 q^{2k+2}}{(\sqrt{Q}-1)^2}\\
        &= \frac{1}{(\sqrt{Q}-1)^2}\cdot \exp{-2(k+1)\log{\frac{1}{q}}}\\
        & = \frac{1}{(\sqrt{Q}-1)^2}\cdot \exp{-2(k+1)\log{(1+\frac{2}{\sqrt{Q}-1})}}\\
        &\geq  \frac{1}{(\sqrt{Q}-1)^2}\cdot \exp{\frac{-4(k+1)}{\sqrt{Q}-1}}.
    \end{align*}
    Take $k = \lfloor \frac{\sqrt{Q}-1}{4}\log \frac{\epsilon_i}{n\epsilon(\sqrt{Q}-1)^2)}\rfloor-1$, since $t\leq k$, we have
    \begin{align*}
        \frac{F_i(\x) - F(\hat{\x})}{F_i(\0) - F(\hat{\x})} \geq \frac{n\epsilon}{\epsilon_i},
    \end{align*}
    which leads to
    \begin{align*}
        n\epsilon &\leq F_i^t(\x) - F_i(\hat{\x}) \\
        &\leq \frac{1}{2}\left(\tilde f_{i,1}(\x) + \tilde f_{i,2}(\x)\right) + \frac{(1-\tilde \mu)(k+\zeta)}{16}c^2-\frac{1}{2}\left(\tilde f_{i,1}(\hat\x) + \tilde f_{i,2}(\hat\x)\right)\\
        &\leq \frac{1}{2}\left(\tilde f_{i,1}(\x) + \tilde f_{i,2}(\x)\right) + \frac{(1-\tilde \mu)(k+1)}{16}c^2-\frac{1}{2}\left(\tilde f_{i,1}(\x^*) + \tilde f_{i,2}(\x^*)\right).
    \end{align*}
    Thus,
    \begin{align*}
        &\quad \left(\tilde f_{i,1}(\x) + \tilde f_{i,2}(\x)\right) - \left(\tilde f_{i,1}(\x^*) + \tilde f_{i,2}(\x^*)\right) \geq 2n\epsilon - \frac{(1-\tilde \mu)(k+1)}{8}c^2.
    \end{align*}
    Take
    \begin{align*}
        c = \min{\left\{\frac{1}{\sqrt{N}}, \sqrt{\frac{8n\epsilon}{(1-\tilde\mu)(k+1)}}, Cq^{k+1}\right\}},
    \end{align*}
    then we have
    \begin{align*}
        &\quad \left(\tilde f_{i,1}(\x) + \tilde f_{i,2}(\x)\right) - \left(\tilde f_{i,1}(\x^*) + \tilde f_{i,2}(\x^*)\right)\geq n\epsilon.
    \end{align*}
    Therefore, if there exists a vector $\v_{ij}$ which holds that $|\<\x,\v_{ij}\>|<\frac{c}{2}$, then $\x$ cannot be an $\epsilon$-optimal point of the hard function $F(\x)$.
\end{proof}

\begin{proof}[Proof of \cor{case1-lowerbound}]

    Our proof relies on some results related to adversary bounds, which are detailed in \append{adversary-method}.

    By \lem{sc-hardfunc-reduction},  of for a hard function $F(\x)$ defined in \defn{case1-hardfunc}, with parameters set by \lem{sc-hardfunc-reduction}, we must output a vector $\x$ such that $\abs{\inner{\x}{\v_{i,j}}} > \frac{c}{2}$ for any $i$ and $j$ to find an $\epsilon$-optimal point of $F(\x)$.

    By \lem{QFCO-MDP-reduction} and \lem{MDP-MCP-reduction}, when $d$ is sufficiently large, for a quantum algorithm that finds $\x$ with the above requirements using $s$ queries, we can easily construct an algorithm solving \prob{MCP} with input size $n\times k$ using $O(s)$ queries. Then, we can use the adversary method to provide a lower bound on the complexity of \prob{MCP}.  Consider \lem{LBofMCP}, any quantum algorithms solving \prob{MCP} must take at least $\Omega (n\sqrt{k})$ queries. This gives a lower bound for case 1 of \prob{QFCO}:

    \begin{align*}
        s = n \sqrt{k} =  \Omega\left(n^{\frac{3}{4}} \left(\frac{1}{\mu}\right)^{\frac{1}{4}} \log^{\frac{1}{2}}{\left(\frac{\Delta \mu}{\epsilon}\right)}\right).
    \end{align*}
    Considering the smoothness parameter $\ell$, we scale the hard function as follows:
    \begin{align*}
        F'(\x) = \frac{1}{\ell} F({\x}),
    \end{align*}
    then $f_i'(\x)$ is 1-lipschitz and $\psi'(\x)$ is $\frac{\mu}{\ell}$-strongly convex, and the $\epsilon$-optimal point of $F(\x)$ is equivalent to the $\frac{\epsilon}{\ell}$-optimal point of $F(\x)$.
    We can complete the proof using a lower bound with the parameter $\ell$:
    \begin{align*}
     \Omega\left(n^{\frac{3}{4}} \left(\frac{\ell}{\mu}\right)^{\frac{1}{4}} \log^{\frac{1}{2}}{\left(\frac{\Delta \mu}{\epsilon \ell}\right)}\right).
    \end{align*}
    Note that this proof also provides a trivial lower bound $\Omega(n)$ since $k$ is greater than 1, and hence our lower bound is
    \begin{align*}
     \Omega\left(n+n^{\frac{3}{4}} \left(\frac{\ell}{\mu}\right)^{\frac{1}{4}} \log^{\frac{1}{2}}{\left(\frac{\Delta \mu}{\epsilon \ell}\right)}\right).
    \end{align*}

    Finally, we calculate the required constraints on the dimension $d$. By strong-convexity $F(\0) \leq F(\x^*) + \frac{\mu}{2}\norm{\x^*}^2$, so $\norm{\x^*} \leq \sqrt{\frac{2\Delta}{\mu}}\coloneqq R$. Since the optimal point lies in the $R$-ball, we restrict the algorithm to query only at points $\x$ such that $\norm{\x} \leq R$. We argue that by a slight modification of the hard instance, Querying points beyond the $R$-ball will not yield additional information. The statement is based on the construction in Appendix C.4 of ~\citet{woodworth2016tight}. Define $f_{i,j}'$ through its gradient as:
    \begin{align*}
        \nabla f_{i,j}'(\x) &=
        \begin{cases}
            \nabla f_{i,j}(\x)  &\norm{\x} \leq R;\\
            \nabla f_{i,j}\left(R\frac{\x}{\norm{\x}}\right) &\norm{\x} \ge R.\\
        \end{cases}\\
        \nabla \psi'(\x) &=
        \begin{cases}
            \nabla \psi(\x)  &\norm{\x} \leq R;\\
            \nabla \psi\left(R\frac{\x}{\norm{\x}}\right) -\mu R\frac{\x}{\norm{\x}} +\frac{\mu}{2} \norm{\x}^2 &\norm{\x} \ge R.\\
        \end{cases}
    \end{align*}
    Note that for the new construcion, $f_i'(\x)$ is continuous and $\ell$-smooth, and $\psi'(\x)$ is still $\mu$-strongly convex. Furthermore, it also has the property that querying the function at a point $\x$ beyond the $R$-ball cannot find more information than querying at $R\frac{\x}{\norm{\x}}$. So we can restrict the algorithm not to query points outside the $R$ ball while still maintaining the same capability. Then, we can use the \lem{QFCO-MDP-reduction} to calculate the requirements for dimension $d$: $n(k+1) + \frac{8R^2}{c^2} \log{(2nkN^3)}  = O\left(  \frac{\Delta \ell^2 }{\mu^2n \epsilon} \log^2{\frac{\Delta\mu}{\epsilon\ell}} \log{\frac{n\ell}{\mu}}\right)$, so by \lem{QFCO-MDP-reduction} we can take $d = \left(  \frac{\Delta \ell^2 }{\mu^2n \epsilon} \log^2{\frac{\Delta\mu}{\epsilon\ell}} \log{\frac{n\ell}{\mu}}\right)$.
\end{proof}


\subsection{Proofs for the smooth and non-strongly convex setting}
\label{append:case2-lb-proof}
\begin{proof}[Proof of \lem{case2-hardfunc-reduction}]

    The proof is inspired by the results from Appendix C.3 in ~\citet{woodworth2016tight}.

    Without loss of gengerality, we can assume that $\ell = R =1$.
    As in \defn{case2-hardfunc}, for the parameters $C$, $k$ and $c$ given later, we define $\frac{n}{2}$ pairs of functions:
    \begin{align*}
    f_{i,1}(\x) &= \frac{1}{16}\big( \inner{\x}{\v_{i,0}}^2 - 2C\inner{\x}{\v_{i,0}} + \sum_{r\text{ even}}^k \phi_c\left(\inner{\x}{\v_{i,r-1}} - \inner{\x}{\v_{i,r}}\right) \big),   \\
    f_{i,2}(\x) &= \frac{1}{16} \big(\phi_c(\inner{\x}{\v_{i,k}}) +\sum_{r\text{ odd}}^k \phi_c\left(\inner{\x}{\v_{i,r-1}} - \inner{\x}{\v_{i,r}}\right) \big)
\end{align*}
where $\v_{ij}$ are orthonormal vectors chosen randomly on the unit sphere in $\R^d$. For the non-strongly convex case, we directly set $\psi(\x)$ to 0.

For $i\in\{1,2,\cdots \frac{n}{2}\}$, define
\begin{align*}
    F_i(\x) &= \frac{1}{2}\big(f_{i,1}(\x) + f_{i,2}(\x)\big) \\
    &= \frac{1}{32} \big( \inner{\x}{\v_{i,0}}^2 - 2C\inner{\x}{\v_{i,0}} + \sum_{r=1}^k \phi_c\left(\inner{\x}{\v_{i,r-1}} - \inner{\x}{\v_{i,r}}\right) + \phi_c(\inner{\x}{\v_{i,k}}) \big)
\end{align*}
then the total function is
\begin{align*}
    F(\x) = \frac{1}{n}\sum_{i=1}^{\frac{n}{2}}\sum_{j=1}^{2}f_{i,j}(\x) = \frac{2}{n} \sum_{i=1}^{\frac{n}{2}} (F_i(\x)).  \\
\end{align*}

To estimate the value of the function $F_i(\x)$, we define
\begin{align*}
    F_i'(\x)
    &= \frac{1}{32} \big( \inner{\x}{\v_{i,0}}^2 - 2C\inner{\x}{\v_{i,0}} + \sum_{r=1}^k \left(\inner{\x}{\v_{i,r-1}} - \inner{\x}{\v_{i,r}}\right)^2 + (\inner{x}{\v_{i,k}})^2 \big).
\end{align*}

Considering the property \eq{helper-function-property} of the helper function, we have
\begin{align}\label{eq:5-3}
    F_i'(\x) - \frac{k+1}{16}c^2 \leq F_i(\x) \leq F_i'(\x).
\end{align}

By first order optimality conditions for $F_i'$, the optimum point $\tilde \x_i$ of $F_i'(\x)$ must satisfy that
\begin{gather*}
    2\<\tilde \x_i , \v_{i,0}\> - \<\tilde \x_i , \v_{i,1}\> = C;\\
    \<\tilde \x_i , \v_{i,j-1}\> -2 \<\tilde \x_i , \v_{i,j}\> + \<\tilde \x_i , \v_{i,j+1}\> = 0 ~~~~\text{for $1\leq j\leq k-1$};\\
    \<\tilde \x_i , \v_{i,k-1}\> - 2 \<\tilde \x_i , \v_{i,k}\> = 0.
\end{gather*}
It can be easily verified that the solution to the above system is
\begin{align*}
    \tilde \x_i = C \sum_{j=0}^k \left(1- \frac{j+1}{k+2}\right) \v_{i,j}
\end{align*}
and the corresponding minimum function value is
\begin{align*}
    F_i'(\tilde\x_i) = -\frac{C^2}{32} \frac{k+1}{k+2}.
\end{align*}
Now, we can calculate the norm of the optimal point $\tilde \x_i$:
\begin{align*}
    \norm{\tilde\x_i}^2 &= C^2 \cdot \sum_{j=0}^{k} \left(1-\frac{j+1}{k+2}\right)^2\\
    &= C^2 \cdot \frac{(k+1)(k+2)(2k+3)}{6(k+2)^2}\\
    &\leq \frac{C^2k}{3},~~~~\text{when } k\geq 1.
\end{align*}

Due to the orthogonality of $\v_{i,j}$, the minimization of the total function $F'(\x)=\frac{2}{n} \sum_{i=1}^{\frac{n}{2}}F_i'(\x)$ can be equivalently represented as the minimization over different orthogonal subspaces $V_i = \operatorname{span}\{\v_{i,0},\v_{i,1},\cdots \v_{i,k}\}$, i.e., the minimization over $\frac{n}{2}$ sub-functions. Therefore, $\tilde\x' \coloneqq \sum_{i=1}^{\frac{n}{2}}\tilde \x_i$ is the optimal point of $F'(\x)$.

Take $C = \sqrt{\frac{6}{nk}}$, we can ensure that $\norm{\sum_{i=1}^{\frac{n}{2}}\tilde \x_i} \leq 1 = R$.

Now, we will bound $F_i'(\x)$ at a point $\x$ such that there exists $q\in\{1,2,\cdots \lfloor k/2 \rfloor\}$ satisfying $\abs{\<\x,\v_{i,q}\>}\leq \frac{c}{2}$. Define
\begin{align*}
    F_i^t(\x)
    &\coloneqq \frac{1}{32} \big( \inner{x}{\v_{i,0}}^2 - 2C\inner{x}{\v_{i,0}} + \sum_{r=1}^t \left(\inner{x}{\v_{i,r-1}} - \inner{x}{\v_{i,r}}\right)^2  \big).
\end{align*}
Observe that $F_i^t(\x)  \leq F_i'(\x)$ for all $\x\in \R^d$. Hence, we can find the minimum value of $F_i^t$ under the condition that $\<\x, \v_{i,t}\>$ is fixed as $b_t$. Similarly, By first order optimality conditions for $F_i^t$, its optimum point $\tilde \x_i^t$ must satisfy that
\begin{gather*}
    2\<\tilde \x_i^t , \v_{i,0}\> - \<\tilde \x_i^t , \v_{i,1}\> = C;\\
    \<\tilde \x_i^t , \v_{i,j-1}\> -2 \<\tilde \x_i^t , \v_{i,j}\> + \<\tilde \x_i^t , \v_{i,j+1}\> = 0 ~~~~\text{for $1\leq j\leq q-1$}.\\
\end{gather*}
The solution to the above system is
\begin{align*}
    \tilde \x_i^t = \sum_{j=0}^{t} \left(\frac{t-j}{t+1}C +\frac{j+1}{t+1} b_t \right) \v_{i,j}, ~~~~\text{where } b_t = \<\x, \v_{i,t}\>,
\end{align*}
and the corresponding minimal value is
\begin{align*}
    F_i^t(\tilde \x^t_i) = \frac{1}{32} \left( -C^2 +\frac{(C-b_t)^2}{t+1} \right).
\end{align*}
Consider the case that $\abs{b_q} = \abs{\<\x, \v_{i,q}\>} \leq \frac{c}{2}$, take $c\leq (2-\sqrt{3})C$, we have
\begin{align}
    F_i'(\x) &\geq F_i^t(\x) \geq F_i^t(\tilde \x_i^t) \notag \\
    &\geq \frac{1}{32} \left( -C^2 +\frac{(C-b_q)^2}{q+1} \right) \notag\\
    &\geq \frac{1}{32} \left( -C^2 +\frac{(C-\frac{c}{2})^2}{\frac{k}{2}+1} \right) \notag\\
    &\geq -\frac{C^2}{32}\left(1-\frac{1}{\frac{2}{3}k+\frac{4}{3}}\right)\notag\\
    &\geq -\frac{C^2}{32}\left(1-\frac{1}{\frac{2}{3}k+1}\right). \label{eq:5-4}
\end{align}

    Combine \eq{5-3} and \eq{5-4}, when $k\geq 3$ and $\abs{\<\x, \v_{i,q}\>} \leq \frac{c}{2}$,
    \begin{align*}
        F_i(\x) - F_i(\x^*) &\geq F_i'(\x) - \frac{(k+1)c^2}{16} - F_i(\tilde \x_i)\\
        &\geq F_i'(\x) - \frac{(k+1)c^2}{16} - F_i'(\tilde \x_i)\\
        &\geq  -\frac{C^2}{32}\left(1-\frac{1}{\frac{2}{3}k+1}\right) + \frac{C^2}{32} \frac{k+1}{k+2} - \frac{(k+1)c^2}{16}\\
        &\geq \frac{C^2}{32}\left(\frac{1}{\frac{2}{3}k+1} -\frac{1}{k+2}\right)\frac{(k+1)c^2}{16} \\
        &\geq \frac{1}{32(k+1)^2n} - \frac{(k+1)c^2}{16}.
    \end{align*}
    When $\epsilon <\frac{1}{4096 n}$, setting $k = \lfloor \frac{1}{16\sqrt{\epsilon n}}\rfloor - 1\geq 3$ and $c = \min{\{\frac{1}{\sqrt{N}}, (2-\sqrt{3})C, 8\sqrt{\frac{\epsilon}{k}} \}}$, we can ensure that
    \begin{align*}
        F_i(\x) - F_i(x^*) \geq 8\epsilon - 4\epsilon = 4\epsilon.
    \end{align*}
    Therefore, if for at least $\frac{n}{4}$ of the $i$'s it holds that $\abs{\inner{\x}{\v_{i,j_i}}} < \frac{c}{2}$ for some $j_i \leq q = \lfloor \frac{k}{2}\rfloor$ , then $\x$ cannot be an $\epsilon$-optimal point of the total function $F(\x)$.
\end{proof}

\begin{proof}[Proof of \cor{case2-lowerbound}]
    \lem{case2-hardfunc-reduction} tells us that to find an $\epsilon$-optimal point of the hard function defined in \defn{case2-hardfunc}, with parameters set by \lem{case2-hardfunc-reduction}, we must output a vector $\x$ such that for at least $\frac{n}{4}$ of the $i$'s it holds that $\abs{\inner{\x}{\v_{i,j}}} > \frac{c}{2}$ for any $j\leq q = \lfloor \frac{k}{2}\rfloor$.

    By \lem{QFCO-MDP-reduction'} and \lem{MDP-MCP-reduction}, when $d$ is sufficiently large, for a quantum algorithm finding $\x$ that satisfies the above requirements using $s$ queries, we can construct an algorithm solving \prob{MCP} with input size $n\times q$ using $O\big((s+\frac{n}{2})\log{n}\big)$ queries. From \lem{LBofMCP}, any quantum algorithms solving \prob{MCP} must take at least $\Omega (n\sqrt{q}) = \Omega(n \sqrt{k})$ queries. Therefore, we obtain a lower bound for case 2 of \prob{QFCO}:

    \begin{align*}
        s = \frac{1}{\log n}\cdot n \sqrt{k} =  \Omega\left(n^{\frac{3}{4}}\frac{1}{\log{n}}\left(\frac{1}{\epsilon}\right)^{\frac{1}{4}} \right).
    \end{align*}
    Take $\ell$ and $R$ into account, we can simply scale the hard function as follows:
    \begin{align*}
        F'(\x) = \frac{1}{l R^2} F(\frac{\x}{R}).
    \end{align*}
    Take $\epsilon' = \frac{\epsilon}{\ell R^2}$, we can finish the proof with a lower bound
    \begin{align*}
     \Omega\left(n^{\frac{3}{4}}\frac{1}{\log{n}}\left(\frac{\ell}{\epsilon}\right)^{\frac{1}{4}} R^{\frac{1}{2}}\right).
    \end{align*}
    Next, we point out that $\Omega(n)$ is a trivial lower bound of \prob{QFCO}, as the quantum algorithm must take at least $\Omega(n)$ queries to get the values of vector of $\frac{n}{4}$ sub-function in orthogonal subspace.
    Taking into account the aforementioned discussion, our lower bound for the smooth and non-strongly convex is
    \begin{align*}
     \Omega\left(n+n^{\frac{3}{4}}\frac{1}{\log{n}}\left(\frac{\ell}{\epsilon}\right)^{\frac{1}{4}} R^{\frac{1}{2}}\right).
    \end{align*}
    Finally, we verify the requirements to dimension $d$: $n(k+1) + \frac{8R^2}{c^2} \log{(2nkN^3)}  = O\left(   \left(\sqrt{\frac{n}{\epsilon}}+\frac{1}{\epsilon^2}\right)\log \left(\frac{n}{\epsilon}\right)\right)$, so by \lem{QFCO-MDP-reduction'} we can take $d = \Omega\left(\left(\sqrt{\frac{n}{\epsilon}}+\frac{1}{\epsilon^2}\right)\log \left(\frac{n}{\epsilon}\right)\right)$.

\end{proof}


\subsection{Proofs for the Lipschitz and non-strongly convex setting}
\label{append:case4-lb-proof}
\begin{proof}[Proof of \lem{case4-hardfunc-reduction}]

    The proof is inspired by the results from Appendix C.1 in ~\citet{woodworth2016tight}.

    Without loss of generality, we can suppose that $m$ is even. Otherwise we can simply set the last sub-function to 0, and the query complexity is reduced by a factor $\frac{m-1}{m}$. As in \defn{case4-hardfunc}, for values $b,c$ and $k$ to be fixed later, we define $\frac{n}{2}$ pairs of sub-functions:
    \begin{gather*}
f_{i,1}(\x) = \frac{1}{\sqrt{2}} |b-\<\x,\v_{i,0}\>|+\frac{1}{2\sqrt{k}} \sum_{r\text{ even}}^k \chi_c\left(\inner{\x}{\v_{i,r-1}} - \inner{\x}{\v_{i,r}}\right),\\
f_{i,2}(\x) = \frac{1}{2\sqrt{k}} \sum_{r\text{ odd}}^k \chi_c\left(\inner{\x}{\v_{i,r-1}} - \inner{\x}{\v_{i,r}}\right).
\end{gather*}
where $\v_{ij}$ are random orthonormal vectors on the unit sphere in $\R_d$.

For convenience,  for $i\in\{1,2,\cdots \frac{n}{2}\}$, we define
\begin{align*}
    F_i(\x) &= \frac{1}{2}\left(f_{i,1}(\x) + f_{i,2}(\x) \right)\\
    &= \frac{1}{2\sqrt{2}} |b - \<\x,\v_{i,0}\>| +\frac{1}{4\sqrt{k}} \sum_{j=1}^k \chi_c(\<\x,\v_{i,j-1}\> - \<\x,\v_{i,j}\>).
\end{align*}
It is straightforward to verify that $F_i(\x)$ takes its minimum value 0 when $\<\x,\v_{i,j}\> = b$ for all $j$. Since each sub-function is defined on orthogonal subspaces, the total function $F(\x)=\frac{1}{n}\sum_{i=1}^{\frac{n}{2}} \left(f_{i,1}(\x) +f_{i,2}(\x)\right)$ is minimized at point
\begin{equation*}
    \x^* = b\cdot \sum_{i=1}^{\frac{n}{2}} \sum_{j=0}^{k} \v_{i,j}
\end{equation*}
and $\x^*$ is also an optimal point of $F_i(\x)$ for all $i$.
We set $b = \sqrt{\frac{2}{n(k+1)}}$, such that $\norm{\x^*} = 1 = R$.

Now we bound $F_i(\x)$ at a point $\x$ when there exists $t\leq k$ such that $\<\x,\v_{i,t}\>\leq \frac{c}{2}$.
\begin{align*}
    F_i(\x) - F_i(\x^*) &\geq F_i(\x) - 0\\
    &\geq \frac{1}{2\sqrt{2}} |b - \<\x,\v_{i,0}\>| +\frac{1}{4\sqrt{k}} \sum_{j=1}^k \left( |\<\x,\v_{i,j-1}\> - \<\x,\v_{i,j}\>| - c\right)\\
    &= \frac{1}{2\sqrt{2}} |b - \<\x,\v_{i,0}\>| +\frac{1}{4\sqrt{k}} \sum_{j=1}^k \left( |\<\x,\v_{i,j-1}\> - \<\x,\v_{i,j}\>|\right) -  \frac{k}{4\sqrt{k}} c\\
    & \geq \frac{1}{2\sqrt{2}} |b - \<\x,\v_{i,0}\>| +\frac{1}{4\sqrt{k}} \sum_{j=1}^t \left( |\<\x,\v_{i,j-1}\> - \<\x,v_{i,j}\>|\right) -  \frac{k}{4\sqrt{k}} c\\
    &\geq \frac{1}{2\sqrt{2}} |b - \<\x,\v_{i,0}\>| +\frac{1}{4\sqrt{k}} |\<\x,\v_{i,0}\> - \<\x,\v_{i,t}\>| -  \frac{k}{4\sqrt{k}} c\\
    & \geq \frac{1}{2\sqrt{2}} |b - \<\x,\v_{i,0}\>| +\frac{1}{4\sqrt{k}} |\<\x,\v_{i,0}\>| -\frac{1}{4\sqrt{k}}\frac{c}{2} -  \frac{k}{4\sqrt{k}} c\\
    &\geq -\frac{2k+1}{8\sqrt{k}}c + \min_{z\in\R}\left(\frac{1}{2\sqrt{2}}|b-z|+\frac{1}{4\sqrt{k}}|z|\right)\\
    &= -\frac{2k+1}{8\sqrt{k}}c + \frac{b}{4\sqrt{k}}\\
    &\geq -\frac{2k+1}{8\sqrt{k}}c + \frac{1}{4k\sqrt{n}}.
\end{align*}

Take $c = \min{\{\frac{1}{\sqrt{N}},\frac{\epsilon}{\sqrt{k}}\}}$ and set $k = \lfloor \frac{1}{10\epsilon \sqrt{n}}\rfloor$, we have
\begin{equation*}
    F_i(\x) - F_i(\x^*) \geq -\frac{\epsilon}{2} + \frac{5}{2}\epsilon = 2\epsilon.
\end{equation*}
Therefore, if for at least $\frac{n}{4}$ i's that there exists $i_j$ such that $\<\x,\v_{i,i_j}\><\frac{c}{2}$, then
\begin{align*}
    F(\x) - F(\x^*) \geq \frac{2}{n} \cdot  \frac{n}{4} \cdot 2\epsilon = \epsilon,
\end{align*}
i.e., $\x$ cannot be an $\epsilon$-optimal point of the hard instance $F$.
\end{proof}

\begin{proof}[Proof of \cor{case4-lowerbound}]
    \lem{case4-hardfunc-reduction} informs us that to find an $\epsilon$-optimal point of the hard function defined in \defn{case4-hardfunc}, with parameters set by \lem{case4-hardfunc-reduction}, we must output a vector $\x$ such that for at least $\frac{n}{4}$ of the $i$'s it holds that $\abs{\inner{\x}{\v_{i,j}}} > \frac{c}{2}$ for any $j$.

    From \lem{QFCO-MDP-reduction'} and \lem{MDP-MCP-reduction}, when $d$ is sufficiently large, for a quantum algorithm finding $\x$ that satisfies the above requirements using $s$ queries, we can construct an algorithm solving \prob{MCP} with input size $n\times k$ using $O\big((s+\frac{n}{2})\log{n}\big)$ queries. Then we consider \lem{LBofMCP}, any quantum algorithms solving \prob{MCP} must take at least $\Omega (n\sqrt{k})$ queries. This gives a lower bound for case 4 of \prob{QFCO}:

    \begin{align*}
        s = \frac{1}{\log n}\cdot n \sqrt{k} =  \Omega\left(n^{\frac{3}{4}}\frac{1}{\log{n}}\left(\frac{1}{\epsilon}\right)^{\frac{1}{2}} \right).
    \end{align*}
    Consider $L$ and $R$, similarly we can scale the hard function as follows:
    \begin{align*}
        F'(\x) = \frac{1}{LR} F\left(\frac{\x}{R}\right).
    \end{align*}
    Take $\epsilon' = \frac{\epsilon}{L R}$, we can finish the proof with a lower bound
    \begin{align*}
     \Omega\left(n+ n^{\frac{3}{4}}\frac{1}{\log{n}}\left(\frac{LR}{\epsilon}\right)^{\frac{1}{2}} \right).
    \end{align*}
    The part $\Omega(n)$ here is similar to \cor{case2-lowerbound}.

    Finally, we verify the requirements to dimension $d$: $n(k+1) + \frac{8R^2}{c^2} \log{(2nkN^3)}  = O\left(\frac{1}{\epsilon^3\sqrt{n}}\log \left(\frac{n}{\epsilon}\right)\right)$, so by \lem{QFCO-MDP-reduction'} we take $d = \Omega\left(\frac{1}{\epsilon^3\sqrt{n}}\log \left(\frac{n}{\epsilon}\right)\right)$.
\end{proof}

\subsection{Proofs for the Lipschitz and strongly convex setting}
\label{append:case3-lb-proof}
Inspired by Appendix C.2 in~\citet{woodworth2016tight}, we now use a reduction from case 4 of \prob{QFCO}, i.e., the Lipschitz and non-strongly convex setting to prove \cor{case3-lowerbound}.

\begin{proof}[Proof of \cor{case3-lowerbound}]

    We use proof by contradiction, assuming that there exists an algorithm $A$ that can find a $\epsilon$-optimal point of the total function in case 3 with $o\left(n+ n^{\frac{3}{4}}\left(\frac{1}{\epsilon\mu}\right)^{\frac{1}{4}} L^{\frac{1}{2}} \frac{1}{\log n}\right)$ quantum queries.

    For a function $F(\x)$ that satisfies case 4, suppose $F(\x) = \frac{1}{n}\sum_{i=1}^{n} f_i(\x) + \psi(\x)$, where $f_i(\x)$ is convex and $L$-lipschitz, $\psi(\x)$ is convex and the optimal point satisfies that $\norm{\x^*}\leq R$. We construct an another function which can be minimized by the algorithm $A$:
    \begin{align*}
        \tilde F(\x) \coloneqq \frac{1}{n}\sum_{i=1}^{n} \tilde f_i(\x) + \tilde \psi(\x),\quad \text{where } \tilde f_i(\x)=f_i(\x) \text{ and } \tilde \psi(\x) = \psi(\x) + \frac{\mu}{2} \norm{\x}^2.
    \end{align*}
    Note that $\tilde f_i(\x)$ is still convex and $L$-lipschitz and $\tilde \psi(\x)$ is $\mu$-strongly convex. Therefore, by assumption, $A$ can find a $\frac{\epsilon}{2}$-optimal point $\hat{\x}$ of $\tilde F(\x)$ with $o\left(n+ n^{\frac{3}{4}}\left(\frac{1}{\epsilon\mu}\right)^{\frac{1}{4}} L^{\frac{1}{2}} \frac{1}{\log n}\right)$ quantum queries.
    Furthermore, set $\mu = \frac{\epsilon}{R^2}$, we have
    \begin{align*}
        F(\x) \leq \tilde F(\x) \leq F(\x) + \frac{\epsilon}{2R^2} \norm{\x}^2 \leq F(\x) + \frac{\epsilon}{2}.
    \end{align*}
    So
    \begin{align*}
        F(\hat{\x}) - F(\x^*) \leq \tilde F(\hat{\x}) - \tilde F(\x^*) + \frac{\epsilon}{2} \leq \frac{\epsilon}{2} +\frac{\epsilon}{2} = \epsilon
    \end{align*}
    which means that $\hat{\x}$ is an $\epsilon$-optimal point of $F(\x)$. The total number of  queries is
    \begin{align*}
        o\left(n+ n^{\frac{3}{4}}\left(\frac{1}{\epsilon\mu}\right)^{\frac{1}{4}} L^{\frac{1}{2}} \frac{1}{\log n}\right) = o\left(n+ n^{\frac{3}{4}}\left(\frac{LR}{\epsilon}\right)^{\frac{1}{2}}\frac{1}{\log n}\right)
    \end{align*}
    which contradicts the conclusion of \cor{case4-lowerbound}. Here we need $\epsilon < \frac{9L^2}{200n\mu}$ and $d = \tilde \Omega\left(\frac{1}{\sqrt{\epsilon^3 n}}\right)$ to achieve the contradiction.
\end{proof}


\end{document}